%% file: ms.tex
\pdfoutput=1
\documentclass[journal]{IEEEtran}
\usepackage{amsmath,units,amssymb,bm,amsthm,siunitx,textcomp}
\newcommand\norm[1]{\left\lVert#1\right\rVert}
\newcommand\curlybrace[1]{\left\{#1\right\}}
\newcommand\bracket[1]{\left(#1\right)}
\DeclareMathOperator{\tr}{tr}
\newtheorem{thm}{Theorem}[section]

\newtheorem{theorem}[thm]{Theorem}
\newtheorem{prop}{Proposition}
\usepackage{epsfig,subcaption}
\usepackage{algorithm, algorithmic} 
\usepackage[flushleft]{threeparttable}
\usepackage{booktabs,caption,adjustbox,realboxes}
\usepackage{diagbox}
\usepackage{multirow,booktabs}
\usepackage{slashbox}
\usepackage{url}
\usepackage{xcolor}

\begin{document}

\title{A Unified Model of Feature Extraction and Clustering for Spike Sorting}
% Optimisation spike sorting model base PCA and K-means

\author{Libo~Huang,~\IEEEmembership{Student~Member,~IEEE,}
        Lu~Gan,~\IEEEmembership{Senior~Member,~IEEE,}
        and~Bingo~Wing-Kuen~Ling,~\IEEEmembership{Senior~Member,~IEEE}% <-this % stops a space         Yan~Zeng,
\thanks{This paper was supported partly by the National Nature Science Foundation of China (no. U1701266, no. 61372173 and no. 61671163), the Team Project of the Education Ministry of the Guangdong Province (no. 2017KCXTD011), the Guangdong Higher Education Engineering Technology Research Center for Big Data on Manufacturing Knowledge Patent (no. 501130144), the Guangdong Province Intellectual Property Key Laboratory Project (no. 2018B030322016) and Hong Kong Innovation and Technology Commission, Enterprise Support Scheme (no. S/E/070/17), and in part by China Scholarship Council (CSC).}
\thanks{L. Huang is with the School of Information Engineering, Guangdong University of Technology, Guangzhou, 510006, China, and  the Department of Electrical and Computer Engineering, Brunel University London, London UB8 3PH, U.K. (e-mail: www.huanglibo@gmail.com).}
\thanks{L. Gan is with the Department of Electronic and Electrical Engineering, Brunel University London, UB8 3PH, U.K. (e-mail: lu.gan@brunel.ac.uk).}
% \thanks{Y. Zeng is with School of Computer Science, Guangdong University of Technology, Guangzhou, 510006, China (e-mail: yanazeng013@gmail.com).}
\thanks{B. W. Ling is with School of Information Engineering, Guangdong University of Technology, Guangzhou, 510006, China (e-mail:yongquanling@gdut.edu.cn).}}

% Department of Electronic and Electrical Engineering, Brunel Univeristy London, UB8 3PH, UK

% \markboth{IEEE Transactions on Neural Systems and Rehabilitation Engineering, ~Vol.~**, No.~*, **~20**}%
\markboth{}
{Huang \MakeLowercase{\textit{et al.}}: ***}

\maketitle

\begin{abstract}
Spike sorting plays an irreplaceable role in understanding brain codes. Traditional spike sorting technologies perform feature extraction and clustering separately after spikes are well detected. However, it may often cause many additional processes and further lead to low-accurate and/or unstable results especially when there are noises and/or overlapping spikes in datasets. To address these issues, in this paper, we proposed a unified optimisation model integrating feature extraction and clustering for spike sorting. Interestingly, instead of the widely used combination strategies, i.e., performing the principal component analysis (PCA) for spike feature extraction and K-means (KM) for clustering in sequence, we unified PCA and KM into one optimisation model, which reduces additional processes with fewer iteration times. Subsequently, by embedding the K-means++ strategy for initialising and a comparison updating rule in the solving process, the proposed model can well handle the noises and/or overlapping interference. Finally, taking the best of the clustering validity indices into the proposed model, we derive an automatic spike sorting method. Plenty of experimental results on both synthetic and real-world datasets confirm that our proposed method outperforms the related state-of-the-art approaches.

\end{abstract}


\input{1.Introduction}

\input{3.Methods}

\input{4.Experiment}
\input{5.Discussion}
\input{6.Appendix}

\bibliographystyle{IEEEtran}
\bibliography{IEEEtran_bibliography}
\end{document}

%% file: 1.Introduction.tex
\section{Introduction} \label{sec:int}
%是什么，包括四个阶段。......背景意义，是很重要的。我们这个工作的必要性。spike sorting .....
%spike sorting 中，前两个阶段比较成熟，researchers focus on 特征提取和据类：for clustering, .....
%，具有哪一些缺点，分析原因。
%针对这些缺点，提出了我们的方法。具体是什么 我们方法的细节。。。。
%我们的贡献。
% While many powerful imaging techniques havebeendeveloped (e.g. membrane voltage, Wu et al., 1994; intrinsic signal, Frostig et al., 1990; fMRI, Ogawa et al., 1992), extracellular recording remains the only method that provides both single neuron and single action potential resolution from large and distributed samples.
%% single electrode spike sorting研究意义
In the neuroscience community, extracellular recording, especially on the single-electrode and the extended tetrodes, remains the only way to provide single action potential resolution to understand the brain codes of large animals or human objects~\cite{lewicki1998review,mokri2017sorting,harris2000accuracy,carlson2019continuing}. Different neurons generate action potentials with a unique extracellular waveform (i.e., ``spike")~\cite{lewicki1998review,rey2015past}. With these unique waveforms, assigning each spike to its tentative neuron becomes possible, and such procedure is called spike sorting~\cite{gibson2011spike,franke2015bayes}.
Outputs from spike sorting are beneficial in broad scopes of applications like neural prosthetic~\cite{linderman2007signal}, brain-machine interfaces~\cite{sanchez2007technology}, treating epilepsy~\cite{gibson2011spike}, etc. 
In real scenarios, since the recorded signals from extracellular implementations consist of mixture local field potentials, noises, spikes, etc., conventional spike sorting techniques firstly have to filter the signals to remove most unwanted noises and disturbances, leaving the filtered trace for the spike detection~\cite{rey2015past}. 
Afterwards, feature extraction and clustering on the detected spikes are needed, which are rather challenging compared with filtering and detection~\cite{mahallati2019cluster}. 
% the fist two straightforward steps as pointed out in~\cite{mahallati2019cluster}.
% . There two steps are far from being satisfactorily settled while 
% Considering time efficiency and visualization, etc., a spike sorting system usually includes feature extraction and clustering after spike filtering and detection. As was pointed out in~\cite{mahallati2019cluster}, the first two steps are relatively mature, the challenging parts lie in the spike feature extraction and clustering.
% rest two parts, i.e., spike feature extraction and clustering, are far from being satisfactorily settled.
% challenging parts lie in the spike feature extraction and clustering. % Generally, we always perform clustering after the spike feature extraction step finished.

Over the past decades, numerous studies have been devoted to extracting discriminative spike features and clustering the featured spikes.
On one hand, the spike feature extraction methods can be fallen into two categories: global-features-based and local-features-based ones. Global-features-based methods include principal components analysis (PCA)~\cite{adamos2008performance}, Fourier transform~\cite{yang2013robust}, wavelet decomposition \cite{hulata2002method,chaure2018novel}, etc. 
Local-features-based methods are those called locality preserving projection (LPP)~\cite{nguyen2014spike}, Laplacian eigenmaps (LE)~\cite{chah2011automated}, and graph-Laplacian (GL)~\cite{ghanbari2010graph}, etc. 
On the other hand, clustering methods are flourishing, e.g., K-means (KM) \cite{chah2011automated,keshtkaran2017noise}, spectral clustering (SC) \cite{nguyen2014spike,huang2020spike}, Gaussian mixture model (GMM) \cite{souza2019spike,keshtkaran2017noise}. Simply combining spike feature extraction and clustering, i.e., separately performing them in sequence, has achieved certain degrees of success.
%Separately performing spike feature extraction and clustering in sequence have achieved certain degrees of success. 
% In particular, it seems that combining any feature extraction methods with any feasible clustering methods for spike sorting can get reasonable results. 
%The global features include the spike time features (or call morphology features) such as spike height and width~\cite{lewicki1998review}, waveform derivatives~\cite{ yang2009improving}, principal components~\cite{adamos2008performance}, and the spike time-frequency features like the coefficients or its derivations of Fourier transform~\cite{yang2013robust}, wavelet~\cite{chaure2018novel}, and wavelet packet decomposition \cite{hulata2002method}, etc. Local features mainly include the features extracted from locality preserving projection (LPP)~\cite{nguyen2014spike}, Laplacian eigenmaps (LE)~\cite{chah2011automated}, and graph-Laplacian (GL)~\cite{ghanbari2010graph}, etc. Separately performing these two stages, spike feature extraction and clustering (e.g., K-means \cite{chah2011automated,keshtkaran2017noise}, spectral clustering \cite{nguyen2014spike,huang2020spike}, Gaussian mixture model (GMM) \cite{souza2019spike,keshtkaran2017noise}), in sequence have achieved certain degrees of success. 

%%% 技术细节（分析原因） 第三段
However, it is unpractical to exhaust all kinds of potential combinations, and what's worse, these combinations may yield unsatisfactory and unstable results for different datasets, especially when there are noises and/or overlapping spikes interference.
% what's worse, the results towards different datasets may vary greatly. That is, they may yield unsatisfactory and unstable results for some datasets, especially when there is noise and/or overlapping spike interference.
% There is no doubt that it is impossible to exhaust all kinds of potential combinations and that the results obtained by the one combination approach on different datasets may vary greatly.
% However, when dealing with different datasets, these methods may not yield satisfactory and stable results, especially when there is noise and/or overlapping spike interference in the datasets.
For instance, suppose we have two noise-disrupted datasets obtained from global-features-based and local-features-based methods as shown in~Fig.\ref{fig:fig1-a} and Fig.\ref{fig:fig1-d}, respectively. For the global-featured spikes in Fig.\ref{fig:fig1-a}, we yield a satisfactory result with KM method (Fig.\ref{fig:fig1-b}), but an unsatisfactory result with SC (Fig.\ref{fig:fig1-c}). Inversely, for the local-featured spikes in Fig.\ref{fig:fig1-d}, opposite conclusions are drawn, i.e., the result of SC is satisfactory (Fig.\ref{fig:fig1-f}) while that of KM method is not (Fig.\ref{fig:fig1-e})\footnote{Note that, for both KM and SC, different initial values will result in different outputs and we only show one of them randomly.}. 
% may obtained global-features-based methods shown in Fig.\ref{fig:fig1-d}\
% when the featured spikes may obtained local-features-based methods distribute in Fig.\ref{fig:fig1-a}, 
% However, inversely, for another distribution of the featured spikes may obtained global-features-based methods shown in Fig.\ref{fig:fig1-d}\footnote{Note that, different initialisation values will result in different outputs and we only show one of them.}, different results are deduced.

Such problems are originated from at least the following three reasons:
1) the internal relationships between feature extraction and clustering methods are not considered when we simply combine them for spike sorting; 2) The features directly extracted from the detected spikes are sensitive to the noise interference from noises and/or overlapping spikes, without considering the feedback from the clustering stage; And 3) unsupervised clustering methods themselves are sensitive to initial values, and may fall into local optimum, etc.
\begin{figure}[t]
	\centering
	\begin{subfigure}{0.32\linewidth}
		\centering
		\includegraphics[width=1\linewidth]{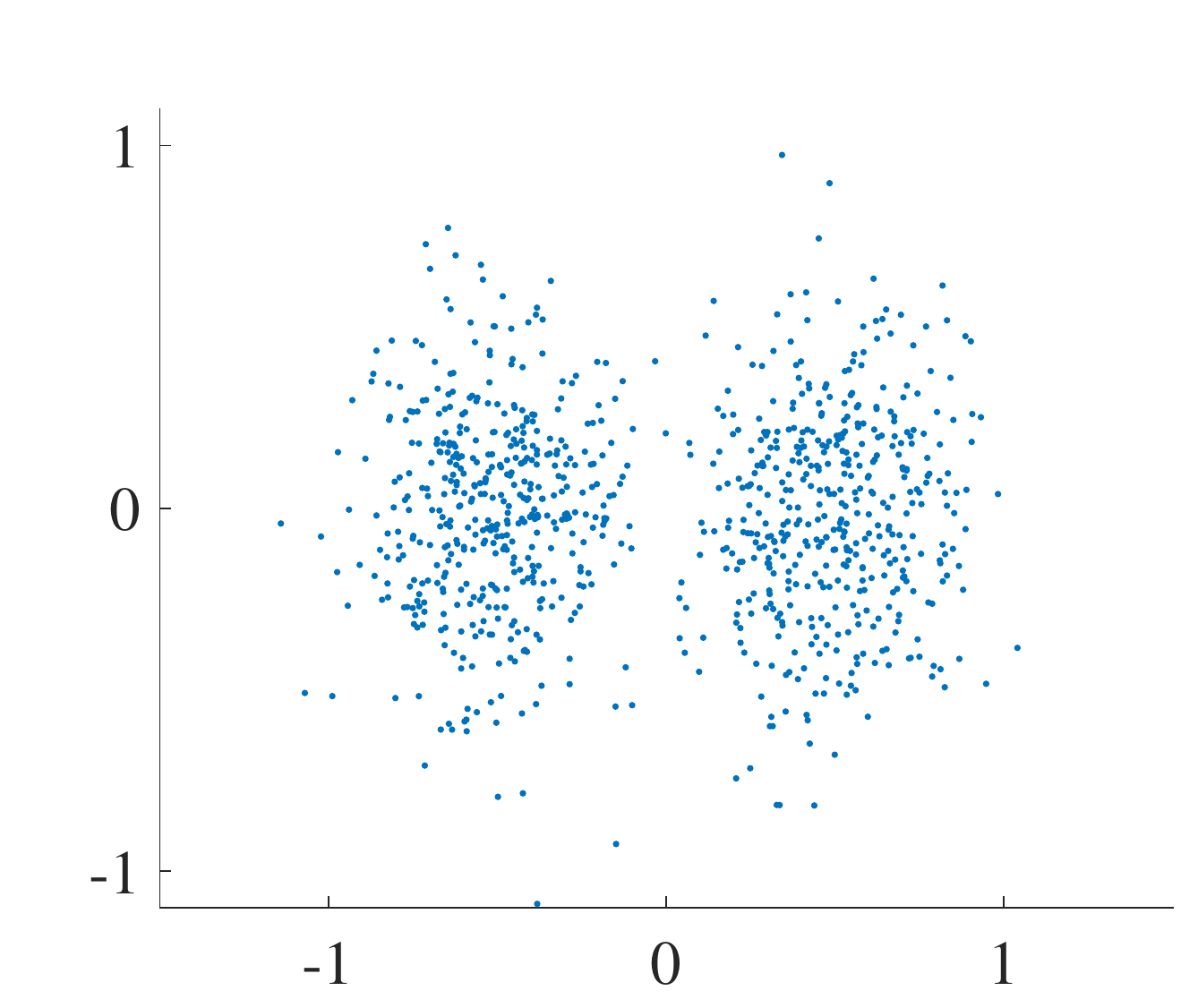}
		\caption{Featured spikes}
		\label{fig:fig1-a}
	\end{subfigure}
	\begin{subfigure}{0.32\linewidth}
		\centering
		\includegraphics[width=1\linewidth]{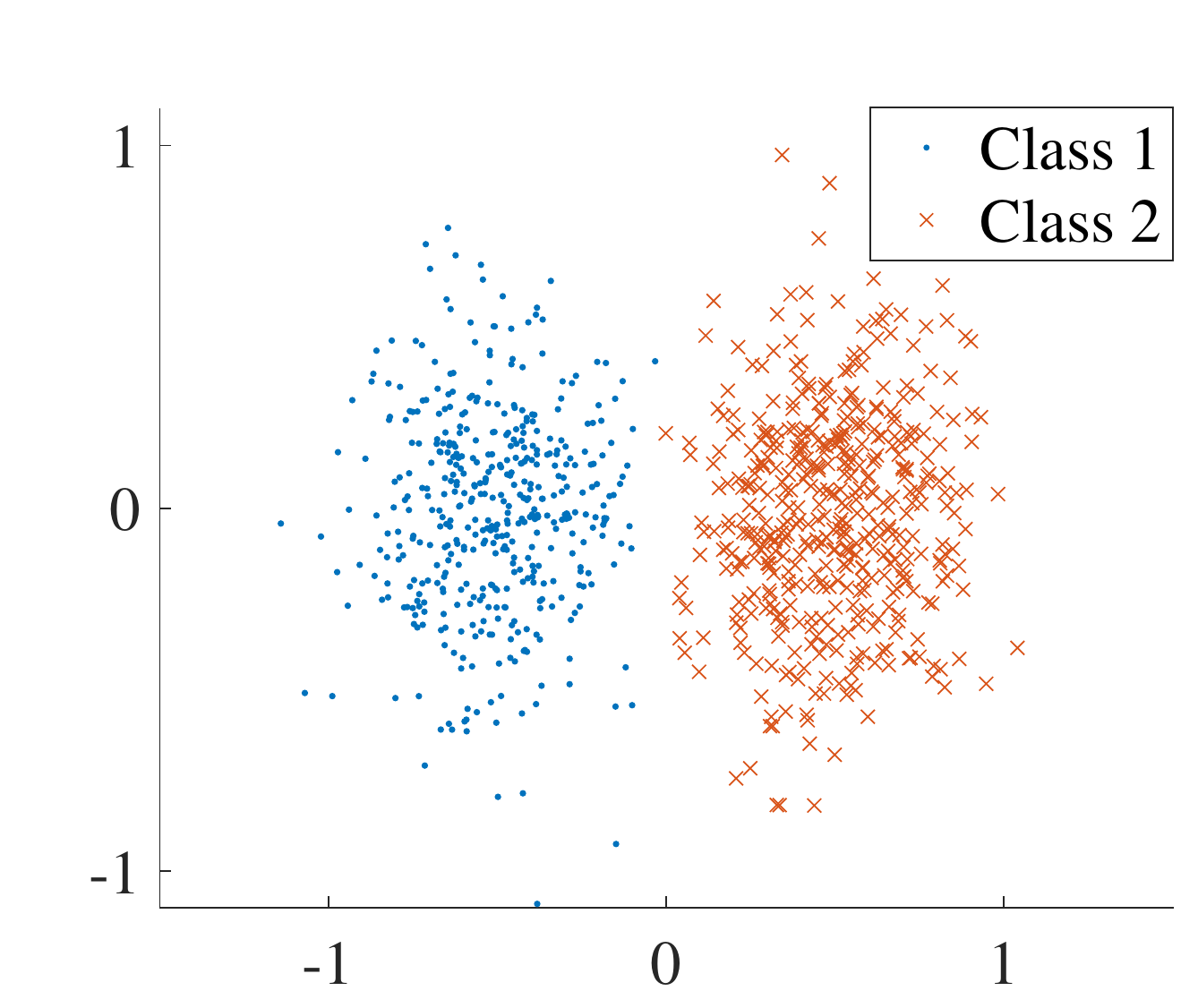}
		\caption{K-means}
		\label{fig:fig1-b}
	\end{subfigure}
	\begin{subfigure}{0.32\linewidth}
		\centering
		\includegraphics[width=1\linewidth]{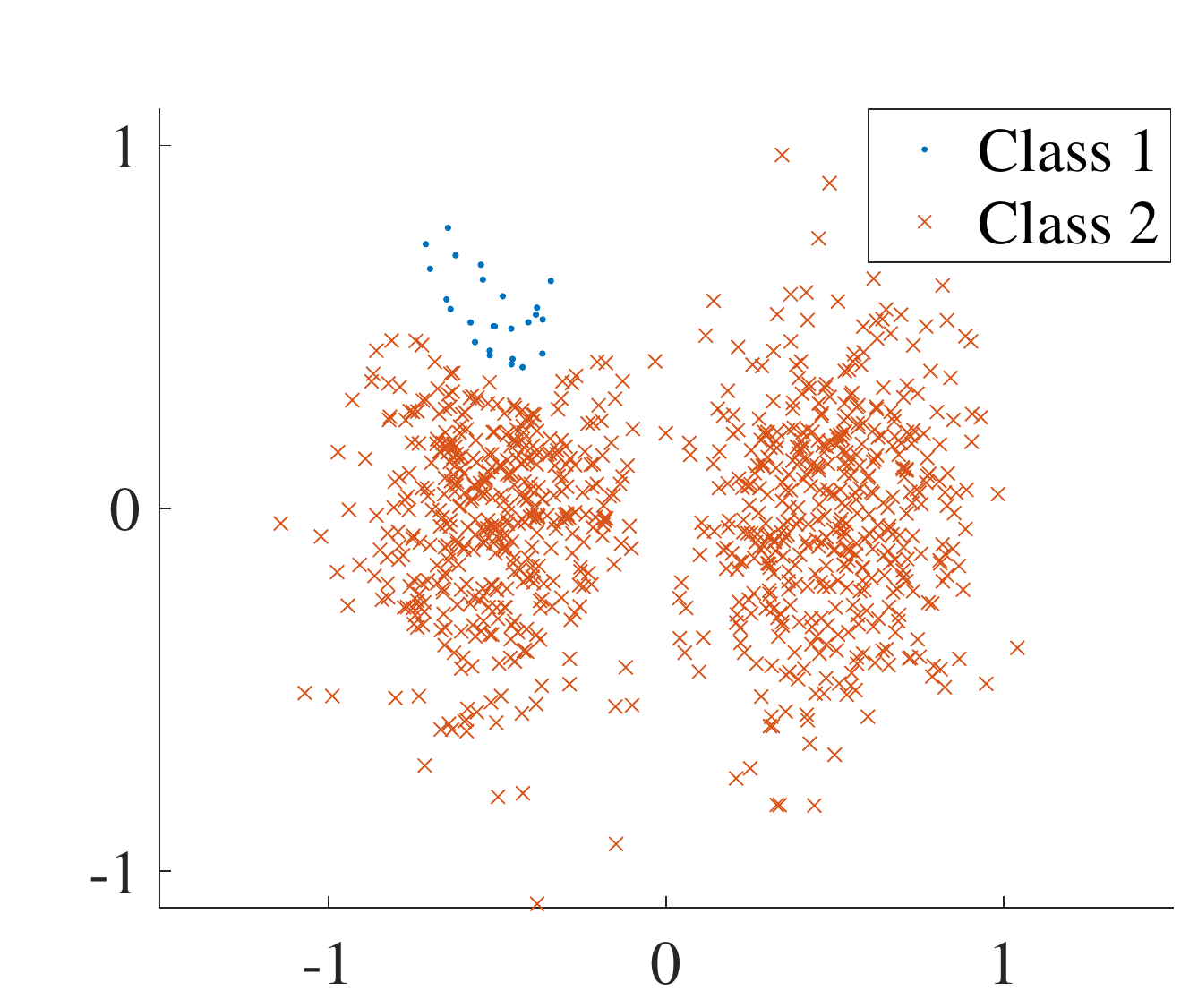}
		\caption{Spectral clustering}
		\label{fig:fig1-c}
	\end{subfigure}
    \newline
    \begin{subfigure}{0.32\linewidth}
		\centering
		\includegraphics[width=1\linewidth]{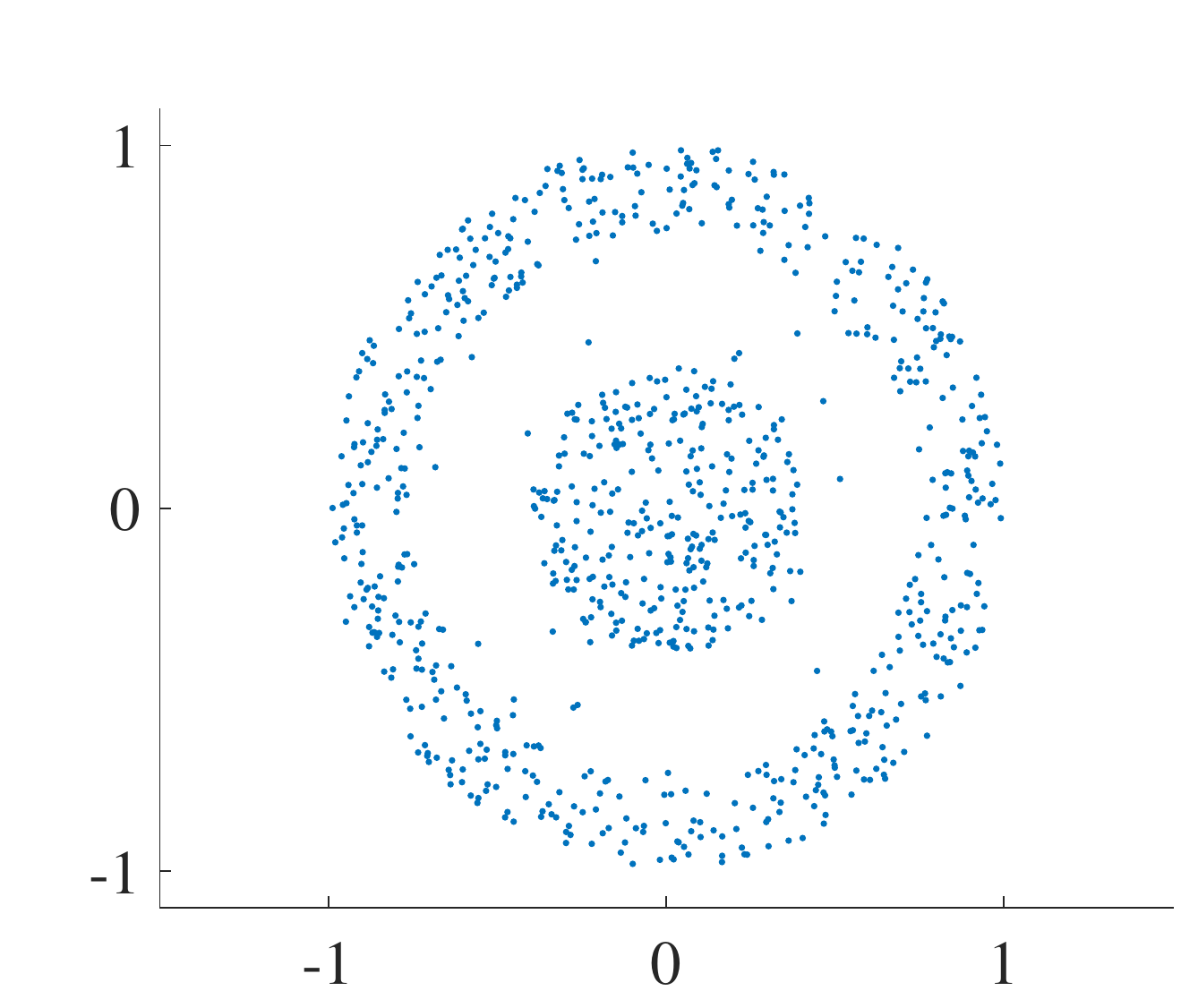}
		\caption{Featured spikes}
		\label{fig:fig1-d}
	\end{subfigure}
	\begin{subfigure}{0.32\linewidth}
		\centering
		\includegraphics[width=1\linewidth]{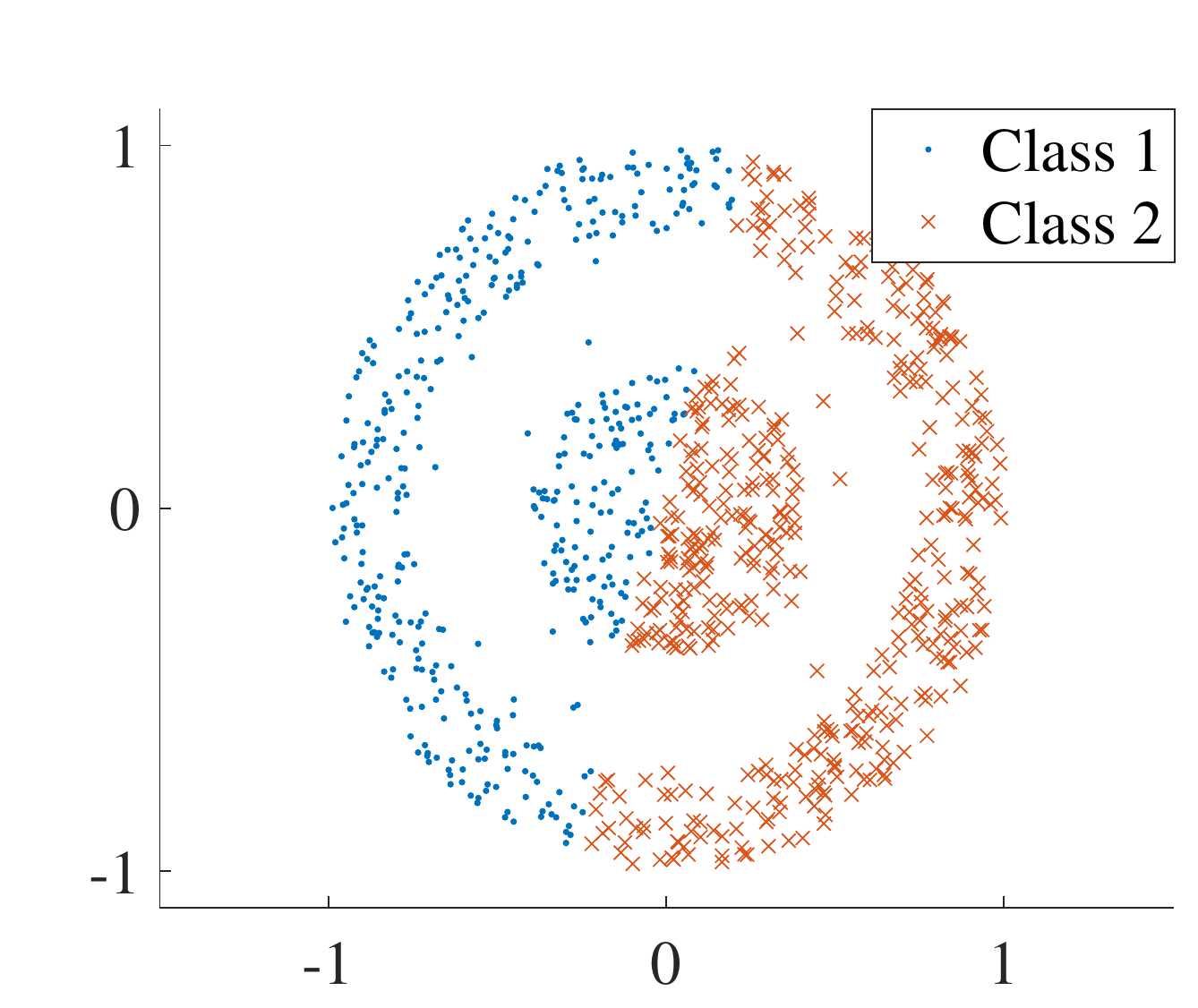}
		\caption{K-means}
		\label{fig:fig1-e}
	\end{subfigure}
	\begin{subfigure}{0.32\linewidth}
		\centering
		\includegraphics[width=1\linewidth]{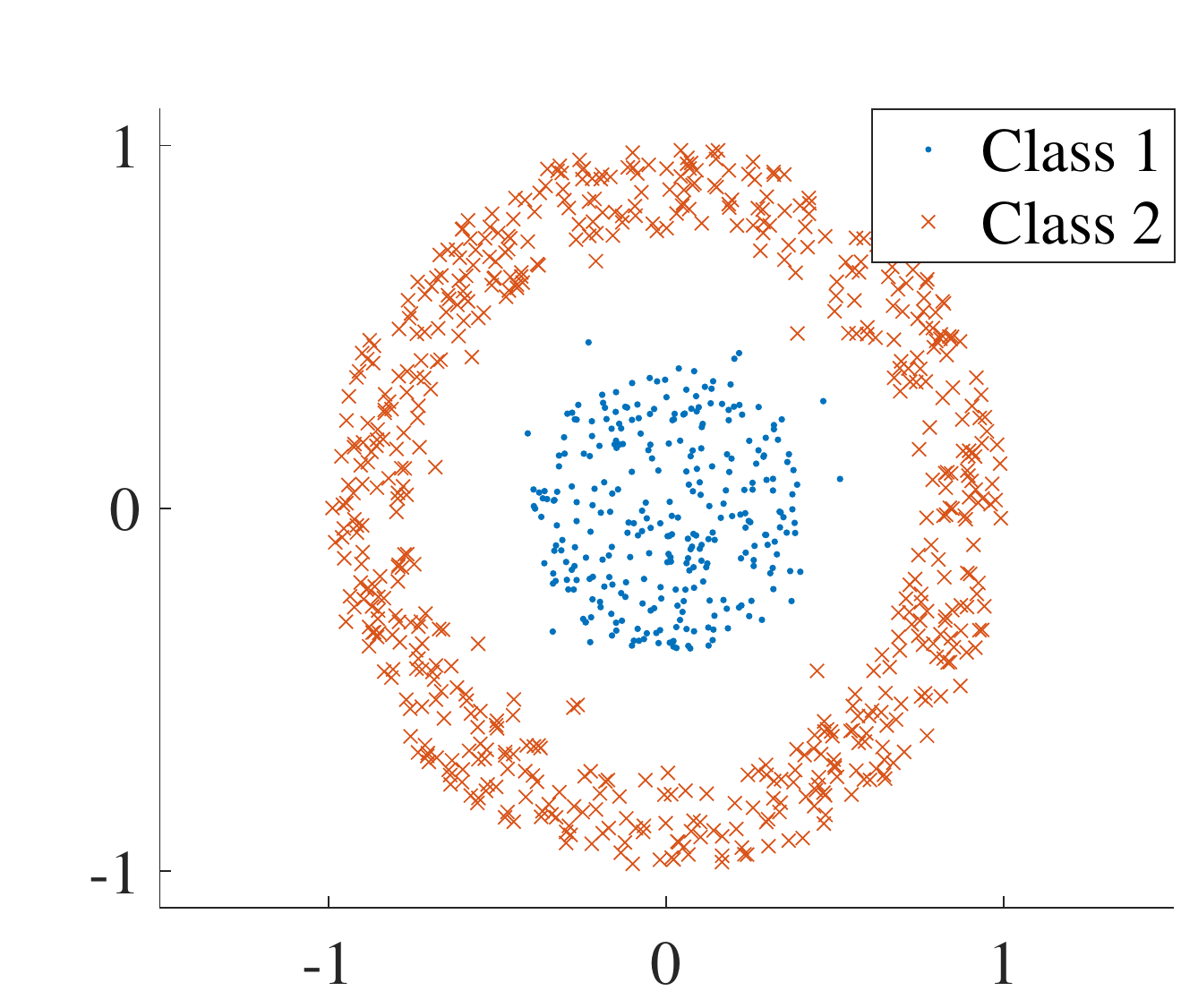}
		\caption{Spectral clustering}
		\label{fig:fig1-f}
	\end{subfigure}
	\caption{(a) and (d) are featured spikes mapped on two dimensions which are supposed to be obtained from global-based and local-based feature extraction methods, respectively. (b) and (e) show the results from K-means while (c) and (f) from spectral clustering. For the spikes distributed in (a), K-means performs better than spectral clustering, while for (d) the conclusion is opposite.}
	\label{fig:fig1}
\end{figure}
% \textcolor{red}{(Some figures in this article are in colour only in the electronic version)}

Recently, Keshtkaran \textit{et al.}~\cite{keshtkaran2014unsupervised,keshtkaran2017noise} introduced linear discriminant analysis guided by K-means (LDAKM)~\cite{ding2007adaptive}, which tried to joint LDA with KM into a coherent framework, to model the feature extraction and clustering problem in spike sorting. Subsequently, they take an additional GMM process after each KM iteration, proposing LDAGMM~\cite{keshtkaran2017noise} to detect the overlapping spikes. % the so-called 
Unfortunately, LDAKM confuses the matrix ratio trace together with trace ratio criteria in the framework, resulting in the outputs vibration or even mis-convergence~\cite{jia2009trace,wang2019unsupervised}. LDAGMM forces more different algorithms to be spliced together compared with LDAKM, which may distract the jointed model. Along with this distraction, more iteration times and computation resources are required. Therefore, both methods still suffer from unstable results.
% Further, it may need computation resource
% may in solving phase and carrying on more . 
% 1) the internal relationships between spike feature extraction and clustering are not considered when we simply combine them for spike sorting; 2) The features directly extracted from detected spikes are sensitive to the noise disturbances, without considering the feedback from the clustering stage; 3) unsupervised clustering methods themselves are sensitive to initial value, and may fall into local optimum, etc. Alternating direction method~\cite{he2012alternating} is engaged in our strategy, which is well studied and guaranteed to converge.

Inspired by these efforts~\cite{keshtkaran2014unsupervised,keshtkaran2017noise}, in this paper, we propose a unified optimisation model of feature extraction and clustering for spike sorting, intrinsically integrating PCA and KM, which enforces these two stages to pursue an identical criteria. In particular, by engaging alternating direction method~\cite{he2012alternating}, we alternatively update the extracted features and clustering results, which means the feedback from the clustering is considered in the feature extraction stage.
% which enforces these two stages uniformity to pursue one objective criteria.  
% which enforces these two stages to be optimised in one identical criteria.
% By engaging alternating direction method~\cite{he2012alternating}, we alternatively update the extracted features and clustering results, facilitating robustness to the noise and taking the feedback from the clustering into account. 
Further, embedding K-means++ initialising and comparison updating strategies in the cluster step, the unstable problem is alleviated significantly. Finally, we derive an automatic spike sorting framework after estimating the number of neurons through existing validity indices.
% The final number of neurals is estimated
% Finally,  we derive an integrated automatic spike sorting framework, using validity indices to estimate the number of neurals.
% Inspired by these efforts~\cite{keshtkaran2014unsupervised,keshtkaran2017noise}, we propose a joint spike feature extraction and clustering optimisation model for spike sorting, intrinsically integrating PCA and K-means, which addresses ..... problems. 
% In particular, how we model, how to solve(one optimisation...not separately optimise....)
% we ....
% the theories of the most used combination strategies, PCA for spike feature extraction and K-means for clustering, try to address these above problems by modelling a joint spike feature extraction and clustering optimisation method for spike sorting. 
% In particular,, unstable problem is solved. 
% To manage these, in this paper, we proposed a unified optimisation model integrated feature extraction and clustering for spike sorting.
% Specifically, the most used combination strategies, principal component analysis (PCA) for spike feature extraction and K-means for clustering in sequence, are well analysed and joined. 
% Further, by embedding new updating strategies in the solving process and by taking advantages of the existing estimation theories, an automatic spike sorting is derived. 
Our contributions are mainly three-fold,
%   LDAGMM simpble outputsly combined different algorithms to get better results, which made the model meaningless [change for anotehr word]. What's worse, taking no consideration on the problems of clustering, unstable outputs still exist. [Make the sentences to be more positive]
%   Inspired by this endeavour, we start from the theories of the most used combination strategies, PCA for spike feature extraction and K-means for clustering, try to address these above problems by modelling a joint spike feature extraction and clustering optimization method for spike sorting. The contributions of this paper are mainly three folds,
\begin{itemize}
	%%方法: 一个自动分类的方法，只讲好处，不解释原因 [Application]
	\item We propose an automatic spike sorting framework which not only keeps the time efficiency but also enjoys the high accuracy property. Moreover, the proposed method guarantees that the results are stable whether noises and/or overlapping spikes exist or not.
	%%理论分析: 时间复杂度分析，solution, 等价证明，时间复杂度 [Theory] 
	% Grammarly 检查拼错的单词 比较好。语法就不要相信了...
	\item We theoretically give the solution to our proposed unified model, and also prove that it can be converted into the PCA or KM like process for each alternate updating. Besides, we provide its computational complexity analysis.
	%结果：单调性验证，stable and accurate results, noise assistance, time-efficiency [Result]
	\item We conduct various experiments on both synthetic and real-world datasets, which verify the monotonicity of our model, rather than vibration. The efficacy and stability of our proposed method are also confirmed. 
	%verify the monotonicity of our model by experimenting on synthesis dataset, meanwhile demonstrate the accurate and stable results of this model along with real datasets.
\end{itemize}
%This paper introduces an idea of unified modelling of spike feature extraction and clustering for spike sorting using PCA and K-means as an example.
%An initial version of this work was presented in (Ekanadham et al., 2011b). A software implementation is available at http://www.cns.nyu.edu/∼lcv/spikeSorting.html~\cite{ekanadham2014unified}.
%The rest of this is constructed as follow.
% This paper introduces an idea of unified modelling of spike feature extraction and clustering for spike sorting using PCA and K-means as an example. 

The rest of this paper is organised as follows. In section~\ref{sec:met}, our automatic spike sorting framework is given in details, particularly, the newly unified feature extraction and clustering model and its solution are formulated and explained together with corresponding theoretical analysis, etc. Section~\ref{sec:exp} presents and discusses our experimental results. Finally, we conclude this paper in Section~\ref{sec:dis}.

%% file: 3.Methods.tex
\section{Methods} \label{sec:met}
In this section, the proposed automatic spike sorting framework is demonstrated in details. It includes filtering, spike detection, and the proposed unified feature extraction and clustering model integrating PCA with KM.

Filtering and amplitude threshold detection in~\cite{chaure2018novel} and ~\cite{ekanadham2014unified} are engaged in our framework. Note that, the particular parameters in these two steps for different datasets are, by default, assigned according to~\cite{chaure2018novel,ekanadham2014unified}, and we will point out these default settings in Section~\ref{sec:exp} as well. Thereafter, each detected spike is stored in the matrix $X\in\mathbb{R}^{d\times n}$ in the form of a column. Here, $d$ and $n$ stand for the number of variables (dimensions) and the number of observations (spikes), respectively.
%\textcolor{red}{Introduce the whole spike sorting process; Emphasize spike filtering and spike detection on two experimental datasets, wave\_clu~\cite{chaure2018novel} and hc-1~\cite{harris2000accuracy,ekanadham2014unified}; Then get the spike matrix $X\in \mathbb{R}^{d\times n}$.}
% Wave_clu
% Spike detection was performed by amplitude thresholding after bandpass filtering the signal (300–6000 Hz, four pole butterworth filter). The threshold (Thr) was automatically set to,
%\begin{equation}
%Thr=4\sigma_n, \ \ \  \sigma_n=median\curlybrace{\frac{|x|}{0.6745}}
%\end{equation}
%where x is the bandpass-filtered signal and σn is an estimate of the stan- dard deviation of the background noise
% Zero-phase filtering is done by using a second-order bandpass elliptic filter in the range of 300–3,000 Hz. Spike detection is performed by setting a threshold as,
%\begin{equation}
%Thr=4\sigma_n, \ \ \  \sigma_n=median\curlybrace{\frac{|x|}{0.6745}}
%\end{equation}
% where x is the bandpass filtered signal (Rey et al. 2015b). 
%\cite{nguyen2014spike}: The first step in the methodology is spike detection, which aims to identify data points that form an action potential. The voltage threshold detection is utilized where the automatic threshold (Thr) is set to:
%\begin{equation}
%Thr=4median\curlybrace{\frac{|x|}{0.6745}}
%\end{equation}https://www.overleaf.com/project/5ef51c284b884200016320cf
\subsection{Theoretically Represent and Optimise PCA and KM}\label{subsec:pca} %Mathematically Analysis of Conventional PCA and KM； Separately Optimise PCA and KM
For easy representation, we first theoretically represent and optimise PCA and KM.% before specifically introducing our proposed joint feature extraction and clustering model, theoretically representing and optimising PCA and KM are needed.
%Before % 具体介绍我们的proposed joint feature extraction and clustering model前，先让我们理论的分析一下PCA和KM.
\subsubsection{PCA}
Given a detected spike dataset $X$, the target of PCA is to construct an orthogonal linear projection matrix $W\in \mathbb{R}^{d\times m}$ so that the projected spikes lie in the $m$ linear spaces with greatest data variances. Here, $m$ is the reduced dimension.% of the spike after feature extraction. %~\cite{wang2019unsupervised}. 
The centralisation matrix is firstly formed by~\cite{wang2019unsupervised}, 
\begin{align} \label{equ:H}
H=I-\frac{1}{n}N\in \mathbb{R}^{n\times n},
\end{align}
where $I$ and $N$ are, respectively, the identity matrix and the all-ones matrix, in which every element is $1$. It can be checked that, $H$ is a symmetric and idempotent matrix~\cite{hou2014discriminative}, i.e., $H=H^T=HH^T$. Then, the spikes can be centred by rightly multiplying $H$, i.e., $XH$. PCA resorts to the total scatter matrix, $XH(XH)^T=XHX^T$, as,
\begin{equation} \label{equ:pca_pro}
\begin{split}
\max_W& \ \ W^TXHX^TW, \\
s.t.,& \ \  W^TW=I.
\end{split}
\end{equation}
% ~\cite{wang2019unsupervised},we need to centre the 
%the matrix $X$ is centred by multiplying  %~\cite{de2006discriminative},
%In fact, $H$ is a symmetric and idempotent matrix, i.e. $H=H^T=HH^T$. PCA resorts to the total scatter matrix $XH(XH)^T=XHX^T$ as~\cite{hou2014discriminative},
Actually, after taking the trace operation $\tr(\cdot)$ on the matrix, $W^TXHX^TW$, the above problem can be solved by performing eigenvalue decomposition (EVD)~\cite{hou2014discriminative} on $XHX^T$. That is $XHX^T=W\Lambda {W}^{-1}$, where $\Lambda$ is the diagonal matrix whose diagonal entries are the corresponding eigenvalues of $XHX^T$. 
% Λ is the diagonal matrix whose diagonal elements are the corresponding eigenvalues, Λii = λi. de2006discriminative, 
Finally, we can get the optimal value and the optimum solution, $W^*$, with the biggest $m$ eigenvalues in $\Lambda$ and the corresponding $m$ eigenvectors in $W$, respectively.

\subsubsection{KM}\label{subsec:kme}
In KM, a label indicator matrix, $G=\curlybrace{G_{ij}|i=1,...,n, j=1,...,c}\in\mathbb{R}^{n\times c}$, is constructed by the rule that $G_{ij}=1$ if the i-$th$ spike belongs to the j-$th$ class and other elements of $G$ are $0$, where $c$ is the number of classes (i.e., neurons). Meanwhile, the centres of $c$ classes are stored as columns in the matrix $M\in\mathbb{R}^{d\times c}$. Then the entire KM clustering process can be expressed as searching for one appropriate local optima solution by performing expectation maximisation (EM) algorithm on the flowing problem~\cite{de2006discriminative, xu2015pca}, 
%同时the centres of the whole $c$ 个类以列的形式存储在矩阵$M\in \mathbb{R}^{d\times c}$中。则整个KM聚类过程可以表示为searching for one appropriate local optima solution by performing expectation maximisation (EM) algorithm on the flowing problem~\cite{de2006discriminative,xu2015pca},
%A label indicator matrix is constructed as $G=\curlybrace{G_{ij}|i=1,...,n, j=1,...,c}\in\mathbb{R}^{n\times c}$, where $G_{ij}=1$ if the i-$th$ data point belongs to the j-$th$ class and other elements of $G$ are $0$. Here $c$ is the total class number~\cite{de2006discriminative}. $M\in \mathbb{R}^{d\times c}$ acts as the centres of the KM clustering with each column of $M$ stands one class centre. KM searches for one appropriate local optima solution by performing expectation maximisation (EM) algorithm~\cite{xu2015pca} on the flowing problem~\cite{de2006discriminative},
\begin{equation} \label{equ:kmeans_pro}
\begin{split}
    \min_{G,M}& \ \ \norm{XH-MG^T}_F^2, \\
    s.t., \ \  G\bm{1}_c=&\bm{1}_n, \ \  G^T\bm{1}_n\succ\bm{1}_c, \ \ G_{ij}\in\curlybrace{0,1},
\end{split}
\end{equation}
where $\norm{X}_F=\left( \sum_{i=1}^d \sum_{j=1}^n |X_{ij}|^2 \right)^{\frac{1}{2}}$ is the Frobenius norm with respect to matrix $X$, $\bm{1_c}\in\mathbb{R}^{c\times 1}$ is an all-ones column vector whose elements are $1$, and $X\succ Y$ indicates $X_{ij}>Y_{ij}$ with element-wise inequality. 
\begin{theorem} \label{the:kmeans} %\cite{wang2019unsupervised}
%\textcolor{red}{rewrite}
Let $X\in \mathbb{R}^{d\times n}$ be the input detected spikes, $\tr(\cdot)$ be the trace operation, and $H\in\mathbb{R}^{n\times n}$ be centralisation matrix, i.e., $H=I-\frac{1}{n}N$. Then Problem (\ref{equ:kmeans_pro}) is equivalent to the following one,
\begin{equation*}
\begin{split}
    \min_G & \ \ \tr\left(XH(I-G(G^TG)^{-1}G^T)HX^T\right), \\
    s.t., & \ \ G\bm{1}_c=\bm{1}_n, \ \ G^T\bm{1}_n\succ\bm{1}_c, \ \ G_{ij}\in\curlybrace{0,1}.
\end{split} 
\end{equation*}
\end{theorem}
\begin{proof}{\ref{the:kmeans}}
Denote $f(G,M)=\norm{XH-MG^T}_F^2$. Since $\norm{X}^2_F=\tr(XX^T)=\tr(X^TX)$, we have,
%$f(G,M)=\norm{XH-MG^T}_F^2$. %One have that $\norm{X}^2_F=Tr(XX^T)=Tr(X^TX)$, therefore
\begin{align*}
    f(G,M)%&=\norm{XH-MG^T}_F^2 \\&
    =\tr\left(XHX^T-2MG^THX^T+MG^TGM^T\right).
\end{align*}
According to $\frac{\partial f(G,M)}{\partial M}=0$ along with $\frac{d }{d X}\tr\bracket{XY^T}=\frac{d }{d X}\tr\bracket{Y^TX}=Y$ and $\frac{d}{dX}\tr\bracket{XYX^T}=X\bracket{Y+Y^T}$, % in the Appendix ~\ref{app:A}, %, and from it can be derived that
we can further derive,
\begin{equation} \label{equ:m}
M=XHG\left(G^TG\right)^{-1}.
\end{equation}
Combining Problem \eqref{equ:kmeans_pro} and~Eq.\eqref{equ:m}, we get the result.% in Theorem \ref{the:kmeans}.
%\begin{align*}
%    f(G,M) &= f(G)\\
%    &=Tr\left(XH(I-G(G^TG)^{-1}G^T)HX^T\right).
%\end{align*}
%So, Theorem \ref{the:kmeans} is proved.
\end{proof}

Theorem \ref{the:kmeans}  indicates that KM can potentially be simplified from the core EM mechanism with iteratively updating two variables to the problem with only one variable needed to be solved.
% Theorem \ref{the:kmeans} indicates that KM potentially be simplified from the core EM mechanism with iteratively updating two variables to the problem with only one variable need to solve.
%定理1表示KM是有潜力从EM迭代更新两个变量的核心机制简化为一个变量的模型求解问题。的核心机制 的核心内部
%Theorem 1 shows that the disentanglement of the representation space is helpful and might also necessary for obtaining less classify error. Then, in the following Theorem 2, we show that the less classify error on the source domain will tighten the error bound at the target domain.
%基于此，我们可以在KM模型的基础上增加一个用于spike降维的更新变量，使得模型求解变量总数不多于2从而保证收敛。~\cite{lin2011linearized,he2012alternating}
%\subsection{Joint Spike Feature Extraction and Clustering}
%此外，PCA是一个针对单个变量，W，的最大化问题而KM有潜力转化为一个单变量,G,最小化问题。
%但是考虑优化方便，没有额外的参数，易于表述等，在本文中我们只关注了迹之比作为我们的目标准则。

\subsection{Unify PCA and KM for Spike Feature Extraction and Clustering}
In the conventional spike sorting procedure, KM is performed after doing PCA on the centralised data matrix, $XH$, for dimension reduction. That is to say, $W^TXH$ is actually fed into KM. % is actually performing on the transformed data matrix $W^TXH$. 
In addition, PCA is a maximisation problem for pursing the single variable, $W$, and KM has the potential to be transformed into a minimisation problem for pursing the single-variable, $G$. By employing Theorem \ref{the:kmeans} and a trace operation on the ratio of matrix~\cite{fukunaga2013introduction}, we can integrate PCA and KM into one optimisation problem as follows\footnote{Generally speaking, there are, at least, four ways, including ratio trace, ratio determinant, trace linear combination and trace ratio, to integrate PCA and KM into one objection criteria as shown in~\cite{fukunaga2013introduction}. But in this paper, we only focus on the trace ratio in view of its convenience for optimising and presenting, etc.}, % without additional parameter, illustration easier, etc.
\begin{equation} \label{equ:model}
\begin{split}
    \max_{G,W}& \ \ \tr\curlybrace{\frac{W^TXHX^TW}{W^TXH(I-G(G^TG)^{-1}G^T)HX^TW}},\\
    s.t.,& \ \ G\bm{1}_c=\bm{1}_n, \ \ G^T\bm{1}_n\succ\bm{1}_c, \ \ G_{ij}\in\curlybrace{0,1}, \\
    & \ \ W^TW=I.
\end{split}
\end{equation}

% After determining our objective function, below are the detailed two-updating steps to solve this problem.  in our strategy
We engage the alternating direction method~\cite{he2012alternating} with two-updating steps to solve the above problem, which guarantees to converge.
\subsubsection{Update G} \label{subsubsct:updateG}
%当更新G时，我们可以忽略与W的相关的约束条件，此时问题\eqref{equ:model}转化为，
When updating $G$, we ignore the constraint associated with $W$ and convert the maximisation problem to a minimisation one. Then Problem \eqref{equ:model} is converted to,
\begin{equation}\begin{split}
    \min_{G}& \ \ \tr\curlybrace{\frac{W^TXH(I-G(G^TG)^{-1}G^T)HX^TW}{W^TXHX^TW}},\\
    s.t.,& \ \ G\bm{1}_c=\bm{1}_n, \ \ G^T\bm{1}_n\succ\bm{1}_c, \ \ G_{ij}\in\curlybrace{0,1}.
\end{split}\end{equation}
According to Theorem \ref{the:kmeans}, solving the above problem equals to performing the following KM-like procedure,
\begin{equation} \label{equ:kmeans_w}
\begin{split}
    \min_{G}& \ \ \norm{\left(W^TXHX^TW\right)^{-\frac{1}{2}}W^TXH-MG^T}_F^2, \\
    s.t.,& \ \  G\bm{1}_c=\bm{1}_n, \ \ G^T\bm{1}_n\succ\bm{1}_c, \ \ G_{ij}\in\curlybrace{0,1},
\end{split}
\end{equation}
where $M=\left(W^TXHX^TW\right)^{-\frac{1}{2}}W^TXHG\left(G^TG\right)^{-1}$. To get a better optimum, we use two strategies as follows,
\begin{itemize}
\item Different from randomly initialising the class centres, we make use of the well-known K-means++~\cite{arthur2006k} initialisation strategy.
\item To ameliorate the local optimum problem, we employ a comparison updating rule~\cite{hou2014discriminative} detailed as follows.%, with its demonstration shown as follows. %similar to.
\end{itemize}
\textit{Updating Rule}: 
%为了方便阐述，我们首先定义，
We begin by defining that, %For ease presentation, w
\begin{equation*}
g(G,M)=\norm{\left(W^TXHX^TW\right)^{-\frac{1}{2}}W^TXH-MG^T}_F^2.
\end{equation*}
Suppose that we obtain the solutions to Problem~\eqref{equ:model}, $W_t^*$ and $G_t^*$ along with the corresponding $M_t^*$, after the $t$-th iteration. Before updating $G$ at the $(t+1)$-th iteration, we firstly initialise $\curlybrace{M_{t+1}^1, M_{t+1}^2, ..., M_{t+1}^a}$ by using K-means++ strategy, and calculate the corresponding $\curlybrace{G_{t+1}^1, G_{t+1}^2, ..., G_{t+1}^a}$. Here, $a$ denotes the number of initialisations, which is fixed to be $10$ in our model. Then we update $G$ with the following criteria,
% 假设在第t次迭代结束时，我们得到关于模型（5）的解为$W_t^*$和$G_t^*$以及对应的$M_t^*$. 在t+1次迭代更新G时，先通过多次（10次在我们的实验中）使用KM++策略初始化得到$\curlybrace{M_{t+1}^1, M^{t+1}^2, ..., M^{t+1}^a}$，并计算得相应的$\curlybrace{G_{t+1}^1, G_{t+1}^2, ..., G_{t+1}^a}$, 这里a表示初始化的次数。然后我们利用如下准则更新G,
\begin{equation} \label{equ:updateG}
G_{t+1}^*=
\begin{cases}
G_{t+1}^i,& \ {\rm if} \ g(G_{t+1}^i,M_{t+1}^i)<g(G_{t}^*,M_{t}^*); \\
G_{t+1}',& \ \ {\rm otherwise},
\end{cases}
\end{equation}
where $i=1,2,...,a$, and $G_{t+1}'$ is defined as the solution to following problem,
\begin{equation} \label{equ:kmeans_}
\begin{split}
\min_{G}& \ \ \norm{\left({W_t^*}^TXHX^T{W_t^*}\right)^{-\frac{1}{2}}{W_t^*}^TXH-M_t^*G^T}_F^2, \\
s.t.,& \ \  G\bm{1}_c=\bm{1}_n, \ \ G^T\bm{1}_n\succ\bm{1}_c, \ \ G_{ij}\in\curlybrace{0,1}.
\end{split}
\end{equation}
%In each step of updating $G$, we using KM++ strategy to initialise $M$ with several times (10 in our experiments) and get the corresponding $G$. For ease presentation, we engage,
%\begin{equation*}
%    g(G,M)=\norm{\left(W^TXHX^TW\right)^{-\frac{1}{2}}W^TXH-MG^T}_F^2
%\end{equation*}

%If one of the related objective function $g(G,M)$ is small than that derived by the previous $G$, we update $G$ by this initialisation. Otherwise, we use \eqref{equ:kmeans_w} to update $G$. Mathematically, in the $t$-th iteration, we derive the optimal $G_t^*$, ${M^t}^*$, and ${W^t}^*$. In the $(t+1)$-th iteration, we 
%can get $\curlybrace{{M^{t+1}}^1, {M^{t+1}}^2, ..., {M^{t+1}}^T}$ and $\curlybrace{{G^{t+1}}^1, {G^{t+1}}^2, ..., {G^{t+1}}^T}$ form KM++ initialisation strategy,  where $T$ is the number of initialisations. We update $G$ by the following rule:
%\begin{equation} \label{equ:updateG}
%{G^{t+1}}^*=
%\begin{cases}
%{G^{t+1}}^k,& \ \ if \ g({G^{t+1}}^k,{M^{t+1}}^k)<g({G^{t}}^*,{M^{t}}^*) \\
%G^*,& \ \ otherwise,
%\end{cases}
%\end{equation}
%where $k=1,2,...,K$, and $G^*$ are defined as the solution of following problem,
%\begin{equation} \label{equ:kmeans_}
%\begin{split}
%    \min_{G}& \ \ \norm{\left({W^t}^TXHX^T{W^t}\right)^{-\frac{1}{2}}{W^t}^TXH-MG^T}_F^2, \\
%    s.t.,& \ \  G\bm{1}_c=\bm{1}_n, \ \ G^T\bm{1}_n\succ\bm{1}_c, \ \ G_{ij}\in\curlybrace{0,1},
%\end{split}
%\end{equation}
\subsubsection{Update W} When updating $W$, viewing $G$ as a constant, Problem \eqref{equ:model} can be recast as,
\begin{equation}\begin{split} \label{equ:updateW}
    \max_{W} &\ \ \tr\curlybrace{\frac{W^TXHX^TW}{W^TXH(I-G(G^TG)^{-1}G^T)HX^TW}},\\
    s.t.,& \ \ W^TW=I.
\end{split}\end{equation}
The above problem is a convex optimisation, which has a closed-form solution. It can be solved by the following proposition.

\begin{prop}  \label{prop:updateW}
The optimal solution to (\ref{equ:updateW}) can be derived by picking up $m$ generalised eigenvectors corresponding to the $m$ largest generalised eigenvalues of the decomposition as follows,
\begin{equation} \label{equ:pro}
    (XHX^T)W=\left[XH(I-G(G^TG)^{-1}G^T)HX^T\right]W\Lambda
\end{equation}
\end{prop}
The proofs of this proposition are provided in Appendix~\ref{app:A}.
% This proposition is widely used in the literature, but for completeness of the text, we also provide one of the proofs of this proposition at Appendix \ref{app:B}.
%The proof of this proposition is detailed at Appendix \ref{app:B}. \textcolor{red}{For completeness of the text~\cite{zelnik2005self}}

Finally, the whole solving algorithm of our proposed model is summarised in Alg.\ref{alg:PCAKM}.
\begin{algorithm}[tb]
	\caption{The algorithm for solving Problem \eqref{equ:model}}
	\label{alg:PCAKM}
	\begin{algorithmic}[1]
		\REQUIRE
		The detected spikes $X\in\mathbb{R}^{d\times n}$, the number of neurons $c$.
		\ENSURE
		Transformation matrix $W^*\in\mathbb{R}^{d\times m}$, the allocation results of spikes $G^*\in\mathbb{R}^{n\times c}$.
		\STATE Initialise $m=c-1$, $t=1$; Define $H\in\mathbb{R}^{n\times n}$, $W^{(0)}\in\mathbb{R}^{d\times m}$, and $G^{(0)}\in\mathbb{R}^{n\times c}$ according to Eqs.\eqref{equ:H}, \eqref{equ:pca_pro}, and \eqref{equ:kmeans_pro}, respectively. 
		\STATE\label{alg1:loop} Update $W_t$ according to Proposition \ref{prop:updateW}.
		\STATE Update $G_t$ according to Eq.\eqref{equ:updateG}.
		\STATE Quit the loop if $G_t=G_{t-1}$, otherwise, $t=t+1$ and go back to step \ref{alg1:loop}.
		\STATE Return $G^*=G_t$ and $W^*=W_t $.
	\end{algorithmic}
\end{algorithm}

\subsection{Other Descriptions of the Unified Model} %Other Configurations of Model. Analysing for Computational Complexity, Neurons Estimation, and Features Determination
In this subsection, the computational complexity of the proposed model is analysed, and schemes of automatically estimating of the number of neurons and determining the dimension of the reduced feature are presented.
% In this subsection, three configurations of the proposed spikes sorting method are analysed. They are the computational complexity analysis of the proposed model, automatic estimation of the number of neurons, and the suggested scheme for automatically determining the dimension of the reduced feature. Finally, algorithm for the whole automatic spike sorting process are summarised.
%In this subsection, 关于the proposed spike sorting method 的三个方面的configuration are analysised. They are 联合锋电位特征降维与聚类模型的计算复杂度分析，自动估计信号中的神经元个数（即锋电位类别数）的建议方案，和自动决定锋电位降维后的特征个数的建议方案。最后我们总结了整个自动锋电位分类过程的伪代码。

\subsubsection{Analysing the Computational Complexity} %Model Computational Complexity Analysis
The computational complexity of our proposed model mainly consists of three parts: the calculation of $W^TXHXW$ when $W$ is known; Using KM to solve submodel~\eqref{equ:kmeans_}; Using generalised EVD to solve submodel~\eqref{equ:pro}. In fact, by reusing the constructed $XHX\in\mathbb{R}^{d\times d}$, the time complexity of constructing $W^TXHXW$ is $O(d^2m+m^2d)$. The computational complexity of KM is $O(nmc)$~\cite{xu2015pca}, while that of the generalised EVD is $O(d^2n)$~\cite{hou2014discriminative}. Here, $d$, $n$, $m$, and $c$ denote the original dimensions of the spike, the total number of spikes, the reduced dimensions of the spike, and the number of spike's categories, respectively. Therefore, the total complexity is $O(b(d^2m+m^2d+nmc+d^2n))$. Generally, it can be further simplified to $O(bd^2n)$, which is linear to the number of spikes, since $n\gg d>m,c$ and $b\approx 10$ hold in the single electrode or tetrode extracellular recorded signal.

\subsubsection{Automatically Estimating the Number of Neurons}
Although many spike sorting methods determine the number of neurons by manually specifying or graphically merging and decomposing clusters~\cite{ghanbari2010graph,ekanadham2014unified,rossant2016spike}, integrating existing estimation methods into the model can reduce a lot of unnecessary human intervention~\cite{keshtkaran2014unsupervised}. To this end, we integrate the estimation of the number of neurons after the model initialisation and before the main iterative optimisation. Within a range of the number of neurons in initial settings, the spike sorting algorithm is pre-run to obtain a series of corresponding index values, and then one reasonable number of neurons is picked up based on the best index value~\cite{yang2013robust,nguyen2015automatic}. Two clustering validity indices, Calinski Harabasz~\cite{zhang2018review} and gap~\cite{nguyen2014spike}, are employed in this paper.

\subsubsection{Automatically Determining the Number of Features}
To efficiently reduce the time consumption and further free human intervention, in our method, automatically determining the dimensions of spikes is adopted. Specifically, $m_0$, as the initially reduced dimensions, only affects the estimation of the number of neurons. After estimating the number of neurons $\tilde{c}$, $m=\tilde{c}-1$ is recommended, where $m$ denotes the dimension of the spike in iteration process. This strategy not only reduces computational complexity in each iteration but avoids the singularity problem~\cite{jia2009trace} in EVD.
%为了有效地估计锋电位个数和降低时间复杂度，在提出的自动spike sorting分类模型中，我们采用了自动决定spikes的特征维度的方式。具体而言，$m_0$作为初始的reduced dimension，只影响到锋电位类别的估计。在确定类别的个数后，$m=c-1$是被推荐的。这里$m$表示模型迭代中的信号维度，$c$表示类别个数。一方面减少了运算量；另一方面避免了singularity problem~\cite{jia2009trace}.
%$m_0$ just related to the determination of the number of class.\textcolor{red}{below are some refer text} The upper bounds on the ranks of the matrices are $c-1$, $n-c$, $n-1$ for $Sb$, $Sw$, $St$ respectively. Note that eq.2 is $tr(Sw)$~\cite{de2006discriminative}.
%最后，Our proposed automatic spike sorting framework are summarized in~Alg\ref{alg:PCAKM_auto}.

Finally, our proposed automatic spike sorting procedure is summarised in Alg.\ref{alg:PCAKM_auto}.
\begin{algorithm}[tb]
	\caption{Proposed automatic spike sorting procedure}
	\label{alg:PCAKM_auto}
	\begin{algorithmic}[1]
		\REQUIRE
		The recorded spike signal; The initial reduced dimension $m_0$.
		\ENSURE
		 Transformation matrix $W^*\in\mathbb{R}^{d\times m}$, the allocation results of spikes $G^*\in\mathbb{R}^{n\times \tilde{c}}$, and the estimated number of neurons $\tilde{c}$.
		\STATE Filter and detect the recorded signal to get the spikes signal matrix $X\in\mathbb{R}^{d\times n}$.
%		Apply \textcolor{red}{the Butterworth bandpass filter with passband located at $300-6,000$ Hz} to the recorded signals to get the filtered data $X_0$.
%		\STATE Perform \textcolor{red}{the automatic threshold detection} on $X_0$ to get the detected spikes matrix signal $X\in\mathbb{R}^{d\times n}$.
		\STATE Initialise $t=1$; Define $H\in\mathbb{R}^{n\times n}$ and $W^{(0)}\in\mathbb{R}^{d\times m_0}$  by solving \eqref{equ:H} and \eqref{equ:pca_pro}, respectively.
		\STATE Estimate the best number of neurons $\tilde{c}$ and $G^{(0)}\in\mathbb{R}^{n\times \tilde{c}}$ simultaneously by solving \eqref{equ:kmeans_pro} with $W_0$ and Calinski Harabasz~\cite{zhang2018review} or gap~\cite{nguyen2014spike}.
		\STATE\label{alg:loop} Update $W^{(t)}$ according to Proposition \ref{prop:updateW}.
		\STATE Update $G^{{(t)}}$ according to \eqref{equ:updateG}.
		\STATE Quit the loop if $G^{(t)}=G^{(t-1)}$, otherwise, $t=t+1$ and go back to step \ref{alg:loop}.
		\STATE Return $G^*=G^{(t)}$, $W^*=W^{(t)}$ and $\tilde{c}$.
	\end{algorithmic}
\end{algorithm}

%% file: 4.Experiment.tex
\section{Experiments} \label{sec:exp}
The proposed automatic spike sorting method is evaluated on both several well-known simulated datasets and the real tetrode-recorded dataset. Further, we compared it with several related state-of-the-art spike sorting methods on different evaluation criteria, including neurons estimation, accuracy, time consumption, stability, etc.
All codes of our experiments are implemented in MATLAB R2020a, which is run on a personal computer with Intel Core i7-8750H 2.20 GHz CPU, 32G RAM).
%	We evaluated our spike sorting method in terms of cluster quality, sorting accuracy, robustness (sensitivity) to noise, dimensionality for reliable clustering, and handling of overlapping spikes. Further we compared the performance of our method with several popular spike sorting methods on simulated and real in vivo data. %Noise-robust unsupervised spike sorting based on discriminative subspace learning with outlier handling. % Tetrode data from rat hippocampus\cite{harris2000accuracy}
%	Windows 10 platform
%	on a PC with Intel Core i7-8750H and run on a PC with Windows 10 platform with 2.21 GHz CPU, 32GB main memory.
%	Our experiments were implemented in MATLAB on a PC with Intel Core i7 and 16 GB RAM.
%	All the codes in the experiments are implemented in MATLAB R2010a and run on a Linux machine with 2.66 GHz CPU, 4GB main memory. 
%	The processing program was run on a personal computer (Intel 2.4GHz, 1GB RAM) using MATLAB7.1.

\subsection{Experimental Datasets and Evaluation Criteria} %Materials
% Before going performing the experiments, we, firstly, introduce the materials and the evaluation criteria engaged. 
% Considering of the lack of ground truth, 
Four challenging synthetic neural datasets in Wave\_clus~\cite{chaure2018novel,quiroga2004unsupervised}, which own the ground truths, are used in our simulations. 
% Taking advantages of the ground truth, four challenging synthetic neural datasets in Wave\_clus~\cite{chaure2018novel,quiroga2004unsupervised} are used as our simulations. That is, 
They are Easy1\_noise*, Easy2\_noise*, Difficult1\_noise*, and Difficult2\_noise*, in which the asterisk stands for the embedded noise standard deviation (i.e., noise level) with respect to the main waveforms. Here, the waveforms were issued by three neurons in neocortex and basal ganglia~\cite{quiroga2004unsupervised}, and the noise was mimicked by randomly selecting the real recorded signal with different amplitudes from different neurons~\cite{chaure2018novel}. Note that, `Easy' reflects that the used waveforms are easy to distinct (i.e., they are different from each other) and vice versa for the `Difficult'. Besides, the noise level of `Easy1' is ranging from 0.05 to 0.4 with the step length 0.05 while those of the rest three datasets are in [0.05, 0.2]. %keshtkaran2017noise,ekanadham2014unified,ghanbari2010graph,

We also engaged the publicly available vivo dataset, HC1~\cite{harris2000accuracy} in our real-world experiment. %as our materials. 
HC1 was recorded from the anesthetized rat hippocampus using an extracellular tetrode, with which one extra intracellular electrode was implanted to a single neuron to obtain its actual spike times (refer to ground truth) for evaluation~\cite{wehr1999simultaneous,ekanadham2014unified}.%~\cite{ekanadham2014unified}.
%	We also applied our method to a portion of publicly available data, recorded from CA1 in anesthetized rat hippocampus (Harris et al., 2000). The data include simultaneous recordings from an extracellular tetrode, and an intracellular electrode that was used to obtain the actual spike times (so-called ground truth) for a single cell. %A unified framework and method for automatic neural spike identification
%	Real in vivo data were also used for evaluation of the algorithms. The data were recorded from the rat hippocampus and are publicly available as HC1 dataset [39, 40]. Four sample recordings d1122101:1, d1122102:1, d1492195:3, and d1492106:4 were used for illustration, whereas one recording d533101:4 was used for quantitative performance evaluation. The dataset d533101:4 contained an intracellular recording from a single neuron and a simultaneous extracellular recording which captured the extracellular spike waveforms from the single neuron as well as an unknown number of other neurons. The intracellular recording gave almost certain spike firing times of the single neuron which were used as the ground truth for performance evaluation. %Noise-robust unsupervised spike sorting based on discriminative subspace learning with outlier handling

To take the best detected spike in simulation datasets and avoid errors from detection, we pick up spikes based on the ground truth for Wave\_clus. Then, with the detected spike matrix $X$, we evaluated the performance of determining the number of neurons, clustering accuracy, time consumption, etc. Especially, to reflect the stability, we conducted $20$ independent trials, which resulted in corresponding means and their standard deviations.
% several repeat sorting experiments (20 times in our experiments) are recommended, and the corrsesponding mean outputs along with their standard deviation are recorded. 
As for the real-world dataset,  HC1, the F1 related criteria, false negative rate (FNR, or called as Miss rate) and false positive rate (FPR), are used. Here, the false negative represents that the predicted cluster is negative but the real class is true, while the false positive means the predicted cluster is positive but the real class is false. FNR and FPR are reasonable enough to reflect the clustering results, and lower values for both of them mean better performances~\cite{ekanadham2014unified,keshtkaran2017noise}.
% Here, the false negative and the false positive have the same reflection that the predicted results were different from the true ones. But the false negative stands the predicted cluster is negative but the real class is true, while the false positive indicates the predicted cluster is positive but the real class is false. FNR and FPR are reasonable enough to reflect the clustering outcomes, and higher values for both of them mean better performances~\cite{ekanadham2014unified,keshtkaran2017noise}.
%both of them with less quantity indicates better results~\cite{ekanadham2014unified,keshtkaran2017noise}.
%	\textcolor{red}{Evaluation Criteria}: %
%	为了避免由spike detection带来的误差，合理利用仿真数据集，
%	Further, we have compared it with some recently proposed state-of-the-art spike sorting works on many evaluation criteria, including neural number determination, accuracy, time consumption, stability of results and so on.
\begin{figure}[!t]
	\centering
	\includegraphics[width=0.45\textwidth]{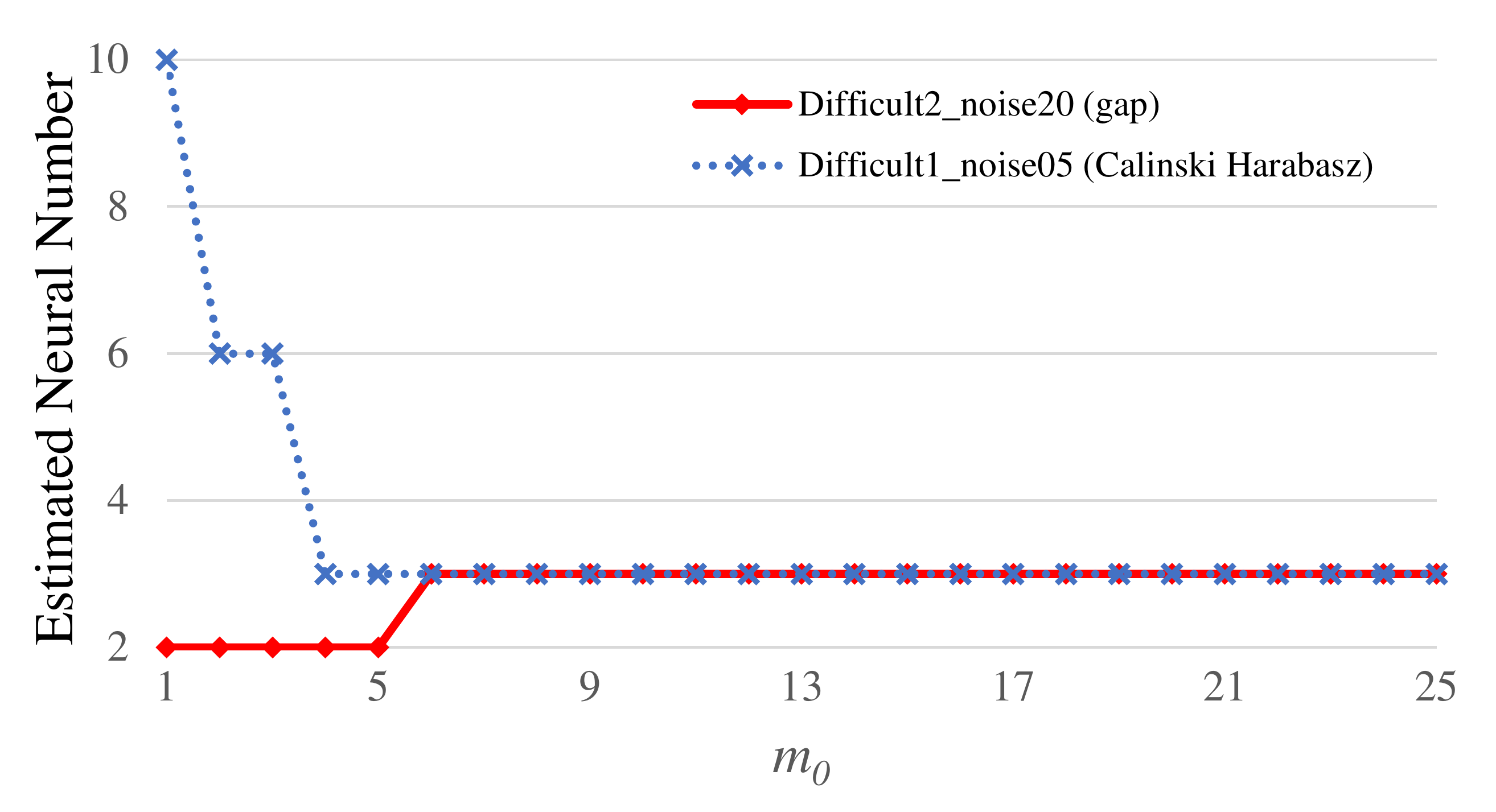}
	\caption{The relationship between the initial number of features $m_0$, and the estimated number of neurons $\tilde{c}$ on two datasets.}
	\label{fig:est_number}
\end{figure}

\subsection{Experiments on Synthetic Dataset - ``Wave\_clus"}
Same as the original works~\cite{chaure2018novel,quiroga2004unsupervised}, ``Wave\_clus" was filtered with passband located at $300-6,000$ Hz firstly. But different from the default configuration using amplitude thresholds to detect spikes, we directly extracted those spikes by using a $64$ sample-length window from the known spike time in ground truths. This facilitates to get consistent spike datasets to perform fair comparison with different methods.

\subsubsection{Determination of the Number of Neurons} \label{subsubsec:number}
\begin{table}[!t]
	\centering
	\caption{The number of estimated neurons using various automatic spike sorting methods on Wave\_clus datasets with (without) overlapping spikes.}
	\label{tab:num_neural}
	\begin{tabular}{@{}ll|lllllll@{}} %|lllllll
		\toprule
		\multirow{2}{*}{DS} & \multirow{2}{*}{NL} & \multicolumn{7}{l}{With Overlapping Spikes (Without)} %& \multicolumn{7}{l}{With Overlapping Spikes}
		\\ \cmidrule(l){3-9} 
		& & A & B & C & D & E
		& P1 & P2 %& A & B & C & D & E & P1 & P2
		\\ \midrule
		Easy1 
		& \begin{tabular}[t]{@{}l@{}}05\\ 10\\ 15\\ 20\\ 25\\ 30\\ 35\\ 40\end{tabular}
		& \begin{tabular}[t]{@{}l@{}}2(3)\\ 3(3)\\ 8(2)\\ 3(2)\\ 2(7)\\ 3(2)\\ 3(3)\\ 4(3)\end{tabular} % A
		& \begin{tabular}[t]{@{}l@{}}2(10)\\ 6(9)\\ 2(4)\\ 3(7)\\ 2(2)\\ 8(3)\\ 4(2)\\ 2(2)\end{tabular} % B
		& \begin{tabular}[t]{@{}l@{}}2(4)\\ 2(4)\\ 2(3)\\ 2(5)\\ 2(3)\\ 2(5)\\ 2(5)\\ 2(4)\end{tabular} % C
		& \begin{tabular}[t]{@{}l@{}}2(3)\\ 2(5)\\ 2(4)\\ 2(5)\\ 2(4)\\ 2(5)\\ 4(5)\\ 2(4)\end{tabular} % D
		& \begin{tabular}[t]{@{}l@{}} 3\\ 3\\ 3\\ 3\\ 6\\ 3\\ 4\\ 3\end{tabular} % E
		& \begin{tabular}[t]{@{}l@{}}3(3)\\ 3(3)\\ 3(3)\\ 3(3)\\ 3(3)\\ 3(3)\\ 3(3)\\ 3(3)\end{tabular} % P1
		& \begin{tabular}[t]{@{}l@{}}3(5)\\ 3(3)\\ 3(3)\\ 3(3)\\ 3(3)\\ 3(3)\\ 3(3)\\ 3(3)\end{tabular} % P2
		\\ \midrule
		Easy2
		& \begin{tabular}[t]{@{}l@{}}05\\ 10\\ 15\\ 20\end{tabular}
		& \begin{tabular}[t]{@{}l@{}}7(4)\\ 3(3)\\ 8(2)\\ 3(2)\end{tabular} %A
		& \begin{tabular}[t]{@{}l@{}}3(7)\\ 7(10)\\ 6(6)\\ 9(10)\end{tabular} %B
		& \begin{tabular}[t]{@{}l@{}}2(5)\\ 2(3)\\ 3(3)\\ 2(3)\end{tabular} %C
		& \begin{tabular}[t]{@{}l@{}}2(2)\\ 2(3)\\ 2(3)\\ 2(4)\end{tabular} %D
		& \begin{tabular}[t]{@{}l@{}}6\\ 3\\ 3\\ 3 \end{tabular} % E
		& \begin{tabular}[t]{@{}l@{}}3(3)\\ 2(3)\\ 2(2)\\ 2(2)\end{tabular} %P1
		& \begin{tabular}[t]{@{}l@{}}3(5)\\ 3(3)\\ 3(3)\\ 3(3)\end{tabular} %P2
		\\ \midrule
		Difficult1
		& \begin{tabular}[t]{@{}l@{}}05\\ 10\\ 15\\ 20\end{tabular} 
		& \begin{tabular}[t]{@{}l@{}}6(7)\\ 3(2)\\ 2(2)\\ 4(3)\end{tabular} %A
		& \begin{tabular}[t]{@{}l@{}}6(10)\\ 2(3)\\ 2(3)\\ 2(2)\end{tabular} %B
		& \begin{tabular}[t]{@{}l@{}}2(4)\\ 2(4)\\ 2(6)\\ 2(5)\end{tabular} %C
		& \begin{tabular}[t]{@{}l@{}}2(4)\\ 2(6)\\ 2(4)\\ 3(7)\end{tabular} %D
		& \begin{tabular}[t]{@{}l@{}}4\\ 3\\ 3\\ 4 \end{tabular} % E
		& \begin{tabular}[t]{@{}l@{}}3(6)\\ 3(3)\\ 3(3)\\ 3(3)\end{tabular} %P1
		& \begin{tabular}[t]{@{}l@{}}8(6)\\ 3(3)\\ 2(2)\\ 2(2)\end{tabular} %P2
		\\ \midrule
		Difficult2
		& \begin{tabular}[t]{@{}l@{}}05\\ 10\\ 15\\ 20\end{tabular}
		& \begin{tabular}[t]{@{}l@{}}8(3)\\ 9(3)\\ 9(3)\\ 8(4)\end{tabular} %A
		& \begin{tabular}[t]{@{}l@{}}3(8)\\ 3(6)\\ 2(10)\\ 2(9)\end{tabular} %B
		& \begin{tabular}[t]{@{}l@{}}2(4)\\ 2(5)\\ 2(3)\\ 3(4)\end{tabular} %C
		& \begin{tabular}[t]{@{}l@{}}2(3)\\ 2(4)\\ 2(3)\\ 2(4)\end{tabular} %D
		& \begin{tabular}[t]{@{}l@{}} 4\\ 3\\ 4\\ 3 \end{tabular} % E
		& \begin{tabular}[t]{@{}l@{}}3(3)\\ 2(3)\\ 2(2)\\ 2(2)\end{tabular} %P1
		& \begin{tabular}[t]{@{}l@{}}3(3)\\ 3(6)\\ 3(3)\\ 2(2)\end{tabular} %P2
		\\ \midrule
		\multicolumn{2}{l|}{Total Hit}
		& 8(7) & 4(3) & 2(6) & 2(5) & 13 & 14(15) & 16(13)
		\\ \bottomrule
	\end{tabular}
	\begin{tablenotes}
		\item DS and NL are the names and the noise levels of datasets, respectively. A, B, C, D, and E stand for the automatic spike sorting frameworks of LE-KM~\cite{chah2011automated}, LPP-SCK~\cite{nguyen2014spike}, LDAKM~\cite{keshtkaran2014unsupervised}, LDAGMM~\cite{keshtkaran2017noise}, and the Wave\_clu~\cite{chaure2018novel} baseline, respectively. P1 and P2 refer to our proposed model with Calinski Harabasz~\cite{zhang2018review} and gap~\cite{nguyen2014spike} index for estimating the number of neurons.
	\end{tablenotes}
\end{table}

\begin{table*}[t!]
	\centering	
	\caption{The average accuracy (\textpm~its standard deviation) and average processing time (\textpm~its standard deviation) of various spike sorting algorithms after repeating 20 times for each dataset without overlapping spikes. Here the number of  classes is fixed to be $3$ and the reduced dimensions of spikes are $2$.}
	\label{tab:acc_without}
	\begin{tabular}{lll|llllllll}
		\toprule
		DS & NL & SN & PCA-KM & \multicolumn{1}{l}{PCA-SCK} & \multicolumn{1}{l}{LE-KM} & LPP-SCK & GF-KM & GF-SCK & LDAKM & Proposed% ($m_0=3$) & LDAGMM 
		\\ \midrule
		Easy1
		& \begin{tabular}[t]{@{}l@{}}05\\ 10\\ 15\\ 20\\ 25\\ 30\\ 35\\ 40\end{tabular}
		& \begin{tabular}[t]{@{}l@{}}2729\\ 2753\\ 2693\\ 2678\\ 2586\\ 2629\\ 2702\\ 2645\end{tabular}
		& \begin{tabular}[t]{@{}l@{}}\textbf{100.00\textpm0.00}\\ \textbf{100.00\textpm0.00}\\ \textbf{100.00\textpm0.00}\\ 99.93\textpm0.00\\ 99.81\textpm0.00\\ 99.20\textpm0.00\\ 97.11\textpm0.00\\ 94.33\textpm0.00 \end{tabular} %PCA-KM
		& \begin{tabular}[t]{@{}l@{}}84.57\textpm15.84\\ 91.27\textpm9.06\\ 85.65\textpm14.98\\ 88.51\textpm11.87\\ 89.24\textpm9.84\\ 98.91\textpm1.05\\ 97.32\textpm0.54\\ 94.66\textpm0.24\end{tabular}%PCA-SCK
		& \multicolumn{1}{l}{\begin{tabular}[t]{@{}l@{}}\textbf{100.00\textpm0.00}\\ 98.95\textpm0.00\\ 93.13\textpm0.00\\ 87.04\textpm0.00\\ 87.39\textpm0.00 \\ 84.33\textpm0.0 \\ 79.24\textpm0.00 \\ 74.43\textpm0.03\end{tabular}} %LE-KM
		& \begin{tabular}[t]{@{}l@{}}83.02\textpm12.92\\ 88.34\textpm8.69\\ 80.75\textpm8.49\\ 80.30\textpm8.01\\ 85.53\textpm1.26\\ 77.60\textpm1.32\\ 69.70\textpm9.85\\ 64.07\textpm8.25\end{tabular} %LPP-SCK
		& \begin{tabular}[t]{@{}l@{}}\textbf{100.00\textpm0.00}\\ \textbf{100.00\textpm0.00}\\ \textbf{100.00\textpm0.00}\\ \textbf{100.00\textpm0.00}\\ \textbf{100.00\textpm0.00}\\ \textbf{100.00\textpm0.00}\\ 99.96\textpm0.00\\\ 97.58\textpm0.00 \end{tabular} %GF-KM
		& \begin{tabular}[t]{@{}l@{}}84.30\textpm16.07\\ 78.62\textpm18.08\\ 88.44\textpm15.01\\ 93.90\textpm8.57\\ 85.33\textpm12.21\\ 88.93\textpm12.30\\ 82.90\textpm19.81\\ 92.57\textpm11.60 \end{tabular} %GF-SCK
		& \begin{tabular}[t]{@{}l@{}}\textbf{100.00\textpm0.00}\\ \textbf{100.00\textpm0.00}\\ \textbf{100.00\textpm0.00}\\ \textbf{100.00\textpm0.00}\\ \textbf{100.00\textpm0.00}\\ \textbf{100.00\textpm0.00}\\ 99.99\textpm0.01\\ 99.99\textpm0.02\end{tabular} %LDA-KM
		%& \begin{tabular}[t]{@{}l@{}}\textbf{100.00\textpm0.00}\\ 97.46\textpm6.06\\ \textbf{100.00\textpm0.00}\\ \textbf{100.00\textpm0.00}\\ \textbf{100.00\textpm0.00}\\ \textbf{100.00\textpm0.00}\\ \textbf{100.00\textpm0.00}\\ \textbf{100.00\textpm0.00}\end{tabular} %%LDA-GMM
		& \begin{tabular}[t]{@{}l@{}}\textbf{100.00\textpm0.00}\\ \textbf{100.00\textpm0.00}\\ \textbf{100.00\textpm0.00}\\ \textbf{100.00\textpm0.00}\\ \textbf{100.00\textpm0.00}\\ \textbf{100.00\textpm0.00}\\ \textbf{100.00\textpm0.00}\\ \textbf{100.00\textpm0.00}\end{tabular} %Proposed
		\\ \midrule
		Easy2
		& \begin{tabular}[t]{@{}l@{}}05\\ 10\\ 15\\ 20\end{tabular}
		& \begin{tabular}[t]{@{}l@{}}2619\\ 2694\\ 2648\\ 2715\end{tabular}
		& \begin{tabular}[t]{@{}l@{}}\textbf{100.00\textpm0.00}\\ 99.07\textpm0.00\\ 92.86\textpm0.00\\ 85.12\textpm0.00 \end{tabular} %PCA-KM
		& \begin{tabular}[t]{@{}l@{}}81.12\textpm15.67\\ 87.15\textpm9.69\\ 78.45\textpm13.84\\ 75.93\textpm11.63 \end{tabular} %PCA-SCK
		& \multicolumn{1}{l}{\begin{tabular}[t]{@{}l@{}}\textbf{100.00\textpm0.00}\\ \textbf{100.00\textpm0.00}\\ \textbf{100.00\textpm0.00}\\ 95.54\textpm0.00\end{tabular}} %LE-KM
		& \begin{tabular}[t]{@{}l@{}}83.44\textpm16.44\\ 89.44\textpm12.66\\ 85.03\textpm17.43\\ 83.65\textpm14.84\end{tabular} %LPP-SCK
		& \begin{tabular}[t]{@{}l@{}}\textbf{100.00\textpm0.00}\\ \textbf{100.00\textpm0.00}\\ \textbf{100.00\textpm0.00}\\ 95.17\textpm0.00 \end{tabular} %GF-KM
		& \begin{tabular}[t]{@{}l@{}}87.49\textpm12.59\\ 89.61\textpm12.36\\ 90.59\textpm15.59\\ 86.24\textpm9.28 \end{tabular} %GF-SCK
		& \begin{tabular}[t]{@{}l@{}}98.60\textpm6.25\\ \textbf{100.00\textpm0.00}\\ \textbf{100.00\textpm0.00}\\ \textbf{100.00\textpm0.00}\end{tabular} %LDA-KM
		%& \begin{tabular}[t]{@{}l@{}}\textbf{100.00\textpm0.00}\\ \textbf{100.00\textpm0.00}\\ \textbf{100.00\textpm0.00}\\ 99.82\textpm0.82\end{tabular} %%LDA-GMM
		& \begin{tabular}[t]{@{}l@{}}\textbf{100.00\textpm0.00}\\ \textbf{100.00\textpm0.00}\\ \textbf{100.00\textpm0.00}\\ \textbf{100.00\textpm0.00}\end{tabular} %Proposed
		\\ \midrule
		Difficult1
		& \begin{tabular}[t]{@{}l@{}}05\\ 10\\ 15\\ 20\end{tabular}
		& \begin{tabular}[t]{@{}l@{}}2616\\ 2638\\ 2660\\ 2624\end{tabular}
		& \begin{tabular}[t]{@{}l@{}}97.13\textpm0.00\\ 88.21\textpm0.00\\ 70.23\textpm0.00\\ 56.59\textpm0.02 \end{tabular}%PCA-KM
		& \begin{tabular}[t]{@{}l@{}}93.80\textpm2.78\\ 86.81\textpm0.66\\ 53.00\textpm6.77\\ 50.26\textpm4.30 \end{tabular}%PCA-SCK
		& \multicolumn{1}{l}{\begin{tabular}[t]{@{}l@{}}\textbf{100.00\textpm0.00}\\ 88.07\textpm0.53\\ 81.99\textpm0.03\\ 60.94\textpm0.00\end{tabular}} %LE-KM
		& \begin{tabular}[t]{@{}l@{}}81.61\textpm19.54\\ 98.35\textpm3.77\\ 80.93\textpm0.51\\ 53.65\textpm4.20\end{tabular} %LPP-SCK
		& \begin{tabular}[t]{@{}l@{}}\textbf{100.00\textpm0.00}\\ 99.70\textpm0.00\\ 81.78\textpm0.05\\ 58.96\textpm0.00 \end{tabular} %GF-KM
		& \begin{tabular}[t]{@{}l@{}}83.93\textpm16.49\\ 93.06\textpm13.44\\ 81.30\textpm0.61\\ 56.71\textpm3.07 \end{tabular} %GF-SCK
		& \begin{tabular}[t]{@{}l@{}}\textbf{100.00\textpm0.00}\\ 99.12\textpm3.94\\ 98.98\textpm4.54\\ 98.86\textpm4.54\end{tabular} %LDA-KM
		%& \begin{tabular}[t]{@{}l@{}}89.75\textpm13.11\\ 91.28\textpm10.96\\ 97.73\textpm7.00\\ \textbf{100.00\textpm0.00}\end{tabular} %LDA-GMM
		& \begin{tabular}[t]{@{}l@{}}\textbf{100.00\textpm0.00}\\ \textbf{100.00\textpm0.00}\\ \textbf{100.00\textpm0.00}\\ \textbf{100.00\textpm0.00}\end{tabular} %Proposed
		\\ \midrule
		Difficult2
		& \begin{tabular}[t]{@{}l@{}}05\\ 10\\ 15\\ 20\end{tabular}
		& \begin{tabular}[t]{@{}l@{}}2535\\ 2742\\ 2631\\ 2716\end{tabular}
		& \begin{tabular}[t]{@{}l@{}}\textbf{100.00\textpm0.00}\\ 99.85\textpm0.00\\ 92.25\textpm0.00\\ 74.11\textpm0.01	\end{tabular} %PCA-KM
		& \begin{tabular}[t]{@{}l@{}}81.30\textpm15.03\\ 95.07\textpm7.54\\ 82.63\textpm10.31\\ 73.77\textpm0.28 \end{tabular} %PCA-SCK
		& \multicolumn{1}{l}{\begin{tabular}[t]{@{}l@{}}\textbf{100.00\textpm0.00}\\ \textbf{100.00\textpm0.00}\\ \textbf{100.00\textpm0.00}\\ \textbf{100.00\textpm0.00}\end{tabular}} %LE-KM
		& \begin{tabular}[t]{@{}l@{}}85.06\textpm16.57\\ 88.15\textpm17.65\\ 87.53\textpm11.63\\ 81.07\textpm14.72\end{tabular} %LPP-SCK
		& \begin{tabular}[t]{@{}l@{}}\textbf{100.00\textpm0.00}\\ \textbf{100.00\textpm0.00}\\ \textbf{100.00\textpm0.00}\\ \textbf{100.00\textpm0.00} \end{tabular} %GF-KM
		& \begin{tabular}[t]{@{}l@{}}84.84\textpm19.04\\ 86.33\textpm14.61\\ 91.17\textpm12.44\\ 84.29\textpm14.31 \end{tabular} %GF-SCK
		& \begin{tabular}[t]{@{}l@{}}\textbf{100.00\textpm0.00}\\ \textbf{100.00\textpm0.00}\\ 98.51\textpm6.67\\ \textbf{100.00\textpm0.00}\end{tabular} %LDA-KM
		%& \begin{tabular}[t]{@{}l@{}}\textbf{100.00\textpm0.00}\\ \textbf{100.00\textpm0.00}\\ \textbf{100.00\textpm0.00}\\ \textbf{100.00\textpm0.00} \end{tabular} %LDA-GMM
		& \begin{tabular}[t]{@{}l@{}}\textbf{100.00\textpm0.00}\\ \textbf{100.00\textpm0.00}\\ \textbf{100.00\textpm0.00}\\ \textbf{100.00\textpm0.00}\end{tabular} %Proposed
		\\ \midrule
		\multicolumn{3}{l|}{Average Time / s}
		& 0.03\textpm0.01 & 0.22\textpm0.01 & 0.20\textpm0.03 & 0.38\textpm0.01 & 0.22\textpm0.03 & 0.43\textpm0.01 & 0.63\textpm0.16 & 0.17\textpm0.14 %& 0.36\textpm0.19
		\\ \bottomrule
	\end{tabular}
\end{table*}
% At first, we compared the capabilities of five automatic spike sorting methods in estimating the number of neurons regardless of whether there are overlapping spikes in the datasets or not. 
Four automatic spike sorting methods (LE-KM~\cite{chah2011automated}, LPP-SCK~\cite{nguyen2014spike}, LDAKM~\cite{keshtkaran2014unsupervised}, LDAGMM~\cite{keshtkaran2017noise}) are compared in estimating the number of neurons on the whole datasets with/without overlapping spikes. Here, SCK is short for landmark-based spectral clustering using KM for landmark-selection~\cite{chen2011large}.
Additionally, the baseline method~\cite{chaure2018novel} performed on Wave\_clus was contrasted on the whole datasets only with overlapping spikes. It should be noted that, since several works~\cite{lefebvre2016recent, pedreira2012many} have emphasised that the number of neurons detected by extracellular electrodes should lie between 8 and 10, in our experiments a more reasonable initial range of the estimated number of neurons is set to [2, 10].
% it is initially set in a more reasonable range, from 2 to 10. Then for LE-KM, LPP-SCK, LDAKM, LDAGMM, and our proposed method with either CalinskiHarabasz~\cite{zhang2018review} or Gap~\cite{nguyen2014spike} index, the initial estimated number of the neurons ranges from $2$ to $10$~\cite{pedreira2012many,lefebvre2016recent}.
% as indicated by the corresponding paper, 
The reduced dimension, $m_0$, at the initialisation stage in our method is assigned to $3$, and the parameters for all compared automatic methods are set by default.
% The parameters for each method are set by default, and for our proposed method, $m_0=3$ is assigned. 
Actually, different values of $m_0$ result in various numbers of estimated neurons but it tends to be stable as $m_0$ increases. Fig.\ref{fig:est_number} shows the inference of $m_0$ to the estimated number of neurons on two datasets without overlapping spikes. Since other datasets have similar results, they are omitted here.

The results of different methods in estimating the number of neurons are presented in Table \ref{tab:num_neural}. Note that the distinct spikes in each dataset are generated by 3 neurons.%, and when the estimated number is 10, the estimation process failed. 
The number in parentheses indicates that the results are deduced on the dataset without the influence from overlapping spikes.
% It can be seen from Table \ref{tab:num_neural} that the results of LE-KM are relatively outperform those of PP-SCK, LDAKM, and LDAGMM. %the outputs of LPP-SCK, LDAKM, and LDAGMM are not satisfactory, and the results of LE-KM are relatively acceptable. 
% The main reasons include, 1) LPP-SCK is sensitive to the reduced dimensions of features, and we tested that changing the default value of this parameter produces greatly different results. 2) For LDAKM and LDAGMM, the estimation process is embedded with iteration optimisation, making the results susceptible to the model's unstable initial point and oscillating iterative process, although it is valid in the original paper when the range of initial estimation interval is [2,5]. On the other hand, our proposed integration of Calinski Harabasz~\cite{zhang2018review} and gap~\cite{nguyen2014spike} index methods, both reach the baseline level. This relies largely on our stable input to the estimation after initialisation.%\footnote{Frankly speaking, this is only a suggested solution of determining the neural number, future improvements are attractive but beyond the scope of this manuscript.}.
%% Reference
It can be seen from Table~\ref{tab:num_neural} that the results of our proposed method reach the baseline level and relatively outperform those of LE-KM, LPP-SCK, LDAKM, and LDAGMM. The main reasons include, 1) the estimation results from LE-KM and LPP-SCK are sensitive to the reduced dimensions of features. 2) For LDAKM and LDAGMM, the estimation process is mixed in the iteration optimisation, making the results susceptible to the model's random initial point and vibrant iteration. 3) Our proposed method integrated with Calinski Harabasz~\cite{zhang2018review} or gap~\cite{nguyen2014spike} index succeed largely with the stable input of the estimation from PCA initialisation as shown in Eq.\eqref{equ:pca_pro}.

\subsubsection{Overall Accuracy and Time Consumption}
Secondly, the overall accuracy and time consumption of various spike sorting methods are compared. The state-of-the-art spike feature extraction methods, PCA~\cite{adamos2008performance}, LE~\cite{chah2011automated}, LPP~\cite{nguyen2014spike}, GF~\cite{ghanbari2010graph}, LDA~\cite{keshtkaran2014unsupervised,keshtkaran2017noise}, along with three automatic spike sorting schemes, LE-KM~\cite{chah2011automated}, LPP-SCK~\cite{nguyen2014spike}, LDAKM~\cite{keshtkaran2014unsupervised,keshtkaran2017noise} are compared. Noted that the results of LDAGMM, which are similar to those of LDAKM, are not shown here due to the space limitation. Further, we also combined either PCA or GF with KM and SCK as spike sorting methods for more general comparisons. 
For the sake of fair comparison, the number of neurons and the reduced dimension for each method are assigned as $3$ and $2$, respectively. 
%which lack considerations of the internal relationships between the spike feature extraction and clustering. 
% It is worth mentioning that, except for specific explanation, all parameters of the model are set by default.
\begin{table*}[!h]
	\centering	
	\caption{The average accuracy (\textpm~its standard deviation) and average processing time (\textpm~its standard deviation) of various spike sorting algorithms after repeating 20 times for each dataset. Here the number of neurons and the reduced dimension are fixed with $3$ and $2$, respectively.}
	\label{tab:acc_with}
	\begin{tabular}{@{}lll|llllllll@{}}
		\toprule
		DS & NL & SN & PCA-KM & PCA-SCK & LE-KM & LPP-SCK & GF-KM & GF-SCK & LDAKM & Proposed (Time) %& LDAGMM 
		\\ \midrule
		Easy1
		& \begin{tabular}[t]{@{}l@{}}05\\ 10\\ 15\\ 20\\ 25\\ 30\\ 35\\ 40\end{tabular}
		& \begin{tabular}[t]{@{}l@{}}3514\\ 3522\\ 3477\\ 3474\\ 3298\\ 3475\\ 3534\\ 3386\end{tabular}
		& \begin{tabular}[t]{@{}l@{}}99.23\textpm0.00\\ 99.40\textpm0.00\\ 99.19\textpm0.00\\ 99.11\textpm0.00\\ 99.00\textpm0.00\\ 98.04\textpm0.00\\ 95.90\textpm0.01\\ 92.98\textpm0.01 \end{tabular} %PCA-KM
		& \begin{tabular}[t]{@{}l@{}}99.24\textpm0.08\\ 99.44\textpm0.03\\ 99.18\textpm0.03\\ 99.02\textpm0.08\\ 98.99\textpm0.08\\ 98.22\textpm0.13\\ 96.28\textpm0.12\\ 92.60\textpm0.40 \end{tabular} %PCA-SCK
		& \begin{tabular}[t]{@{}l@{}}67.59\textpm0.05\\ 98.38\textpm0.00\\ 96.87\textpm0.00\\ 94.67\textpm0.00\\ 90.72\textpm0.00\\ 86.42\textpm0.00\\ 80.48\textpm0.00\\ 73.72\textpm0.00\end{tabular} %LE-KM
		& \begin{tabular}[t]{@{}l@{}}96.20\textpm0.71\\ 97.91\textpm0.41\\ 94.93\textpm0.89\\ 89.04\textpm3.29\\ 81.35\textpm3.42\\ 71.43\textpm3.29\\ 79.55\textpm0.47\\ 70.06\textpm0.61\end{tabular} %LPP-SCK
		& \begin{tabular}[t]{@{}l@{}}99.35\textpm0.00\\ 99.55\textpm0.00\\ 99.65\textpm0.00\\ 99.77\textpm0.00\\ 99.42\textpm0.00\\ 99.45\textpm0.00\\ 99.35\textpm0.00\\ 98.49\textpm0.00 \end{tabular} %GF-KM
		& \begin{tabular}[t]{@{}l@{}}99.31\textpm0.11\\ 99.55\textpm0.11\\ 99.62\textpm0.03\\ 99.57\textpm0.04\\ 90.50\textpm14.54\\ \textbf{99.70\textpm0.04}\\ 99.61\textpm0.02\\ 99.53\textpm0.15\end{tabular} %GF-SCK
		& \begin{tabular}[t]{@{}l@{}}97.85\textpm7.30\\ 99.66\textpm0.14\\ 99.65\textpm0.09\\ 99.77\textpm0.06\\ 99.64\textpm0.06\\ 99.60\textpm0.02\\ 99.55\textpm0.05\\ 99.69\textpm0.07\end{tabular} %LDA-KM
		%& \begin{tabular}[t]{@{}l@{}}99.46\textpm0.10\\ 99.75\textpm0.03\\ 99.65\textpm0.05\\ 99.78\textpm0.03\\ 99.69\textpm0.06\\ 99.59\textpm0.03\\ 99.53\textpm0.04\\ 99.73\textpm0.06\end{tabular} %LDAGMM
		& \begin{tabular}[t]{@{}l@{}}\textbf{99.52\textpm0.00} (0.24)\\ \textbf{99.77\textpm0.00} (0.24)\\ \textbf{99.68\textpm0.00} (0.25)\\ \textbf{99.80\textpm0.00} (0.24)\\ \textbf{99.76\textpm0.00} (0.23)\\ 99.60\textpm0.00 (0.39)\\ \textbf{99.63\textpm0.00} (0.25)\\ \textbf{99.79\textpm0.00} (0.29)\end{tabular} %Proposed
		\\ \midrule
		Easy2
		& \begin{tabular}[t]{@{}l@{}}05\\ 10\\ 15\\ 20\end{tabular}
		& \begin{tabular}[t]{@{}l@{}}3410\\ 3520\\ 3411\\ 3526\end{tabular}
		& \begin{tabular}[t]{@{}l@{}}98.09\textpm0.00\\ 95.94\textpm0.00\\ 90.00\textpm0.00\\ 82.99\textpm0.01\end{tabular} %PCA-KM
		& \begin{tabular}[t]{@{}l@{}}97.82\textpm0.16\\ 95.49\textpm0.63\\ 87.45\textpm0.84\\ 82.68\textpm0.84 \end{tabular} %PCA-SCK
		& \begin{tabular}[t]{@{}l@{}}98.80\textpm0.00\\ 99.26\textpm0.00\\ 98.59\textpm0.00\\ 93.85\textpm0.00\end{tabular} %LE-KM
		& \begin{tabular}[t]{@{}l@{}}98.57\textpm0.04\\ 99.18\textpm0.02\\ 98.59\textpm0.11\\ 93.81\textpm0.84\end{tabular} %LPP-SCK
		& \begin{tabular}[t]{@{}l@{}}99.12\textpm0.00\\ 99.74\textpm0.00\\ 99.38\textpm0.00\\ 91.98\textpm0.01 \end{tabular} %GF-KM
		& \begin{tabular}[t]{@{}l@{}}99.06\textpm0.18\\ 99.75\textpm0.03\\ 99.36\textpm0.15\\ 91.06\textpm0.91\end{tabular} %GF-SCK
		& \begin{tabular}[t]{@{}l@{}}96.97\textpm8.27\\ 98.52\textpm5.49\\ 98.37\textpm5.71\\ 99.60\textpm0.06\end{tabular} %LDA-KM
		%& \begin{tabular}[t]{@{}l@{}}99.69\textpm0.05\\ 99.79\textpm0.03\\ 99.71\textpm0.05\\ 99.64\textpm0.04\end{tabular} %LDAGMM
		& \begin{tabular}[t]{@{}l@{}}\textbf{99.74\textpm0.00} (0.30)\\ \textbf{99.77\textpm0.00} (0.40)\\ \textbf{99.79\textpm0.00} (0.33)\\ \textbf{99.74\textpm0.00} (0.41)\end{tabular} %Proposed
		\\ \midrule
		Difficult1
		& \begin{tabular}[t]{@{}l@{}}05\\ 10\\ 15\\ 20\end{tabular}
		& \begin{tabular}[t]{@{}l@{}}3383\\ 3448\\ 3472\\ 3414\end{tabular}
		& \begin{tabular}[t]{@{}l@{}}95.71\textpm0.00\\ 86.06\textpm0.01\\ 68.32\textpm0.05\\ 55.58\textpm0.03 \end{tabular} %PCA-KM
		& \begin{tabular}[t]{@{}l@{}}94.94\textpm1.20\\ 85.48\textpm0.78\\ 50.52\textpm4.08\\ 47.27\textpm3.05 \end{tabular} %PCA-SCK
		& \begin{tabular}[t]{@{}l@{}}97.37\textpm0.00\\ 97.51\textpm0.00\\ 76.37\textpm0.04\\ 48.22\textpm0.19\end{tabular} %LE-KM
		& \begin{tabular}[t]{@{}l@{}}98.44\textpm0.09\\ 97.35\textpm0.47\\ 71.61\textpm7.47\\ 45.97\textpm2.65\end{tabular} %LPP-SCK
		& \begin{tabular}[t]{@{}l@{}}97.87\textpm0.00\\ 69.37\textpm0.00\\ 79.52\textpm0.01\\ 47.51\textpm0.00 \end{tabular} %GF-KM
		& \begin{tabular}[t]{@{}l@{}}97.77\textpm0.12\\ 96.33\textpm2.34\\ 77.07\textpm4.94\\ 48.54\textpm3.73 \end{tabular} %GF-SCK
		& \begin{tabular}[t]{@{}l@{}}95.57\textpm10.08\\ 96.52\textpm6.86\\ 97.53\textpm2.15\\ 98.42\textpm0.08\end{tabular} %LDA-KM
		%& \begin{tabular}[t]{@{}l@{}}84.68\textpm10.84\\ 70.25\textpm6.72\\ 82.57\textpm13.18\\ 91.26\textpm11.95\end{tabular} %LDAGMM
		& \begin{tabular}[t]{@{}l@{}}\textbf{99.35\textpm0.00} (0.39)\\ \textbf{99.07\textpm0.00} (0.41)\\ \textbf{98.96\textpm0.01} (0.56)\\ \textbf{98.79\textpm0.02} (1.01)\end{tabular} %Proposed
		\\ \midrule
		Difficult2
		& \begin{tabular}[t]{@{}l@{}}05\\ 10\\ 15\\ 20\end{tabular}
		& \begin{tabular}[t]{@{}l@{}}3364\\ 3462\\ 3440\\ 3493\end{tabular}
		& \begin{tabular}[t]{@{}l@{}}98.54\textpm0.00\\ 97.66\textpm0.00\\ 86.65\textpm0.01\\ 68.31\textpm0.00 \end{tabular} %PCA-KM
		& \begin{tabular}[t]{@{}l@{}}98.46\textpm0.03\\ 97.76\textpm0.05\\ 86.28\textpm0.50\\ 67.67\textpm0.81 \end{tabular} %PCA-SCK
		& \begin{tabular}[t]{@{}l@{}}99.67\textpm0.00\\ 99.77\textpm0.00\\ 99.62\textpm0.07\\ 99.63\textpm0.00\end{tabular} %LE-KM
		& \begin{tabular}[t]{@{}l@{}}99.67\textpm0.07\\ 99.76\textpm0.02\\ 99.49\textpm0.03\\ 99.56\textpm0.03\end{tabular} %LPP-SCK
		& \begin{tabular}[t]{@{}l@{}}99.67\textpm0.00\\ 99.62\textpm0.00\\ 99.68\textpm0.00\\ 99.74\textpm0.00 \end{tabular} %GF-KM
		& \begin{tabular}[t]{@{}l@{}}99.66\textpm0.06\\ 99.71\textpm0.06\\ 99.63\textpm0.05\\ 99.73\textpm0.03\end{tabular} %GF-SCK
		& \begin{tabular}[t]{@{}l@{}}99.71\textpm0.04\\ 99.86\textpm0.03\\ 99.76\textpm0.02\\ 99.85\textpm0.06\end{tabular} %LDA-KM
		%& \begin{tabular}[t]{@{}l@{}}99.72\textpm0.06\\ 99.86\textpm0.02\\ 99.76\textpm0.02\\ 99.86\textpm0.03 \end{tabular} %LDAGMM
		& \begin{tabular}[t]{@{}l@{}}\textbf{99.79\textpm0.00} (0.45)\\ \textbf{99.86\textpm0.00} (0.31)\\ \textbf{99.80\textpm0.00} (0.32)\\ \textbf{99.89\textpm0.00} (1.10)\end{tabular} %Proposed
		\\ \midrule
		\multicolumn{3}{l|}{Average Time / s} &	0.04\textpm0.02 & 0.27\textpm0.01 & 0.31\textpm0.03 & 0.52\textpm0.02 & 0.36\textpm0.02 & 0.62\textpm0.03 & 0.64\textpm0.14 & 0.41\textpm0.24 %& 0.39\textpm0.21 
		\\ \bottomrule
	\end{tabular}
\end{table*}

%'Replicates' — Number of times to repeat clustering using new initial cluster centroid positions

Specifically, in LE-KM, the number of nearest neighbours is, by default, set to be $12$, however, the number of times to randomly repeat KM is not specified, and we set it to be $10$ (it is equal to the number of K-means++ initialisations in our proposed method as mentioned in subsection~\ref{subsubsct:updateG}). Besides, the points of landmarks in SCK and the number of neighbours for each sample in LPP are, by default, fixed to be $1000$ and $5$, respectively. As for LDAKM, the open-source code is publicly accessed~\cite{ldagmm2017code}, while the configurations of our proposed method are detailed in Alg.\ref{alg:PCAKM}~\footnote{Please visit the link, \url{https://github.com/HLBayes/Unified-SS}, to access the Matlab codes about the proposed method in this paper.}. Note that, $20$ independent times are performed in each method on each dataset and the final average results along with its standard deviation are recorded.
%Note that, for reflecting the stability of all outputs, every method is repeated $20$ times on each dataset and the final average results along with its standard deviation are recorded.
In our experiments, whether the overlapping spikes exist or not in the datasets is compared.

As shown in Table~\ref{tab:acc_without}, without overlapping spikes, nearly all methods could get acceptable results on datasets with small noise levels. But when the noise level increases, it seems that only the model-based method (i.e., LDAKM and our proposed method) can maintain high accuracy, while other combined feature extraction and clustering in sequence methods failed. On the other hand, the results of methods combined with SCK (i.e., PCA-SCK, LPP-SCK, GF-SCK) and LDAKM are not stable. For example, the standard deviations of the average accuracy for LDAKM are around $5\%$ and those for the SCK-based methods are even more than $15\%$ in the datasets of Easy2\_noise05, Difficult1\_noise10, etc. The standard deviations of the methods combined with KM (i.e., PCA-KM, LE-KM, GF-KM) are comparable with setting a sufficient number of initialisation, however, the average accuracies are not satisfactory. On the contrary, by initialising with K-means++ strategy and updating with comparison rule whose details are shown in section~\ref{subsubsct:updateG}, our proposed method obtains the near-perfect accuracies. At the same time, the processing time is also less than others. What's more interesting is that by calculating the standard deviation of the processing time, the minimum processing time of our proposed method is equal to the simple combined method, PCA-KM. This means the proposed method can automatically adjust the processing time according to the complexity of the dataset. These findings are also consistent with the presence of overlapping spikes as shown in Table~\ref{tab:acc_with}, in which the exact processing time of our method on each dataset is manifested. % and confirms our findings.
\begin{figure*}[t]
	\centering
	\begin{subfigure}{0.30\linewidth}
		\centering
		\includegraphics[width=1\linewidth]{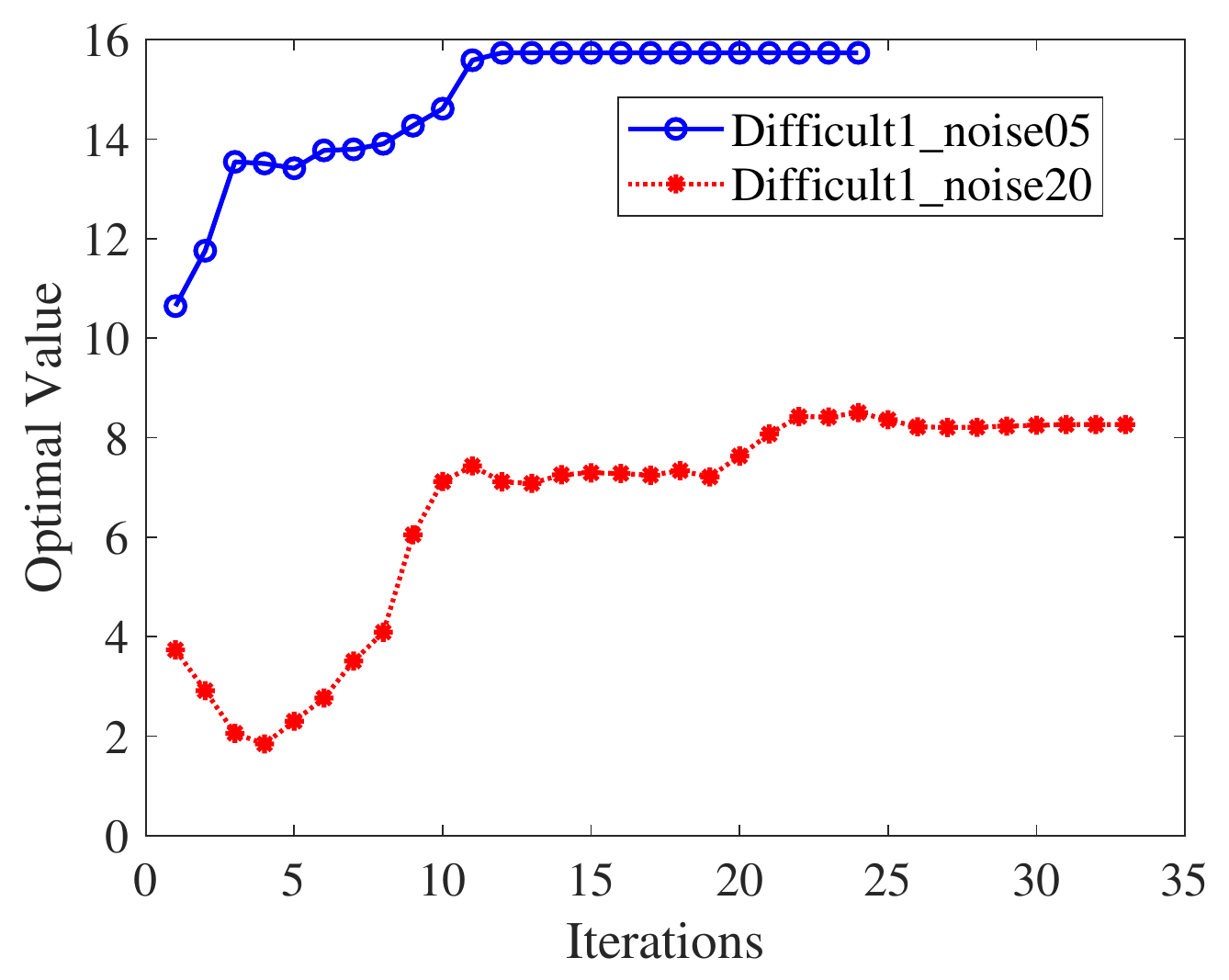} %,height=4.0cm
		\caption{LDAKM}
		\label{fig:valueLDAKM}
	\end{subfigure}
	\begin{subfigure}{0.30\linewidth}
		\centering
		\includegraphics[width=1\linewidth]{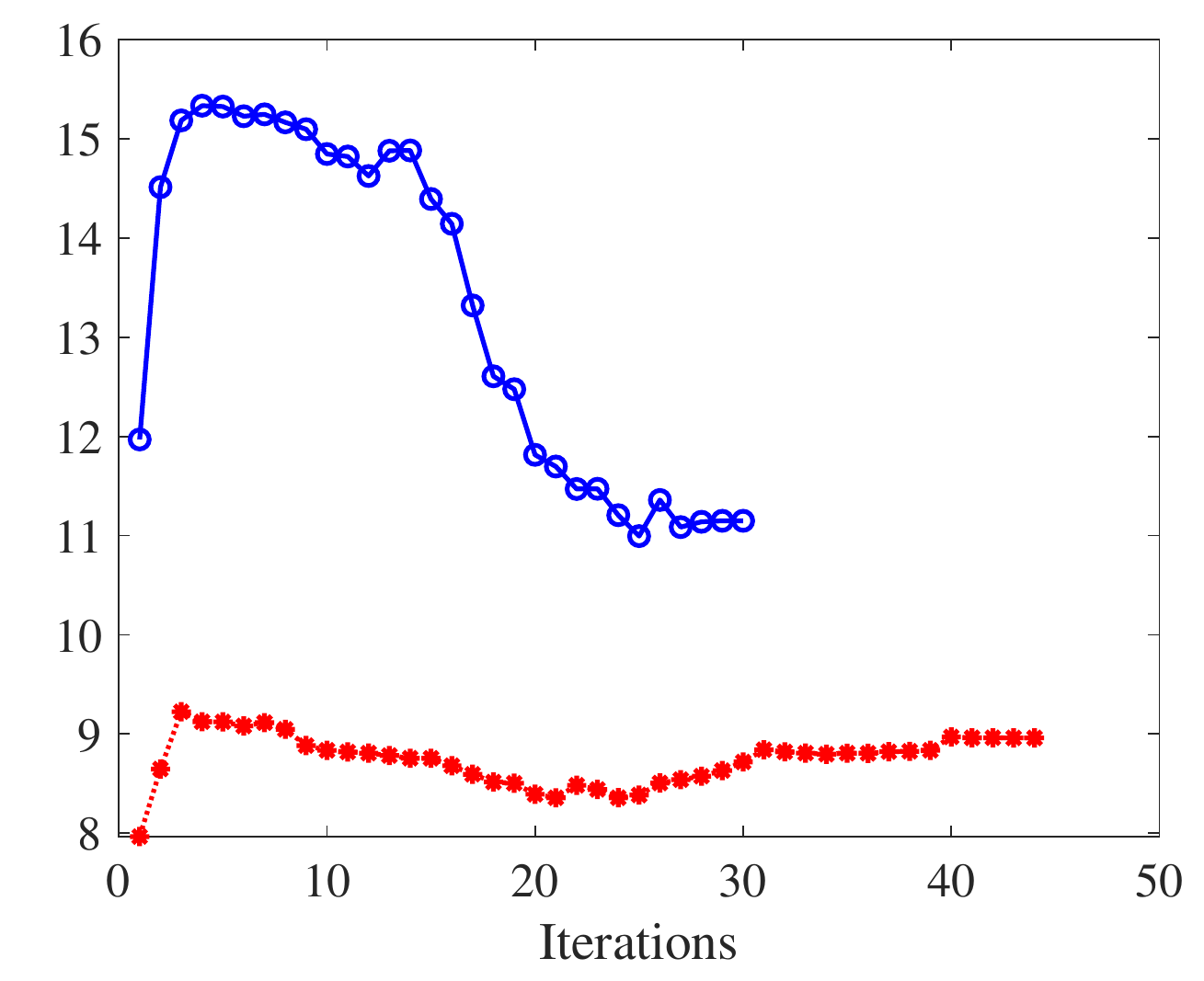} %,height=4.0cm
		\caption{LDAGMM}
		\label{fig:valueLDAGMM}
	\end{subfigure}
	\begin{subfigure}{0.30\linewidth}
		\centering
		\includegraphics[width=1\linewidth]{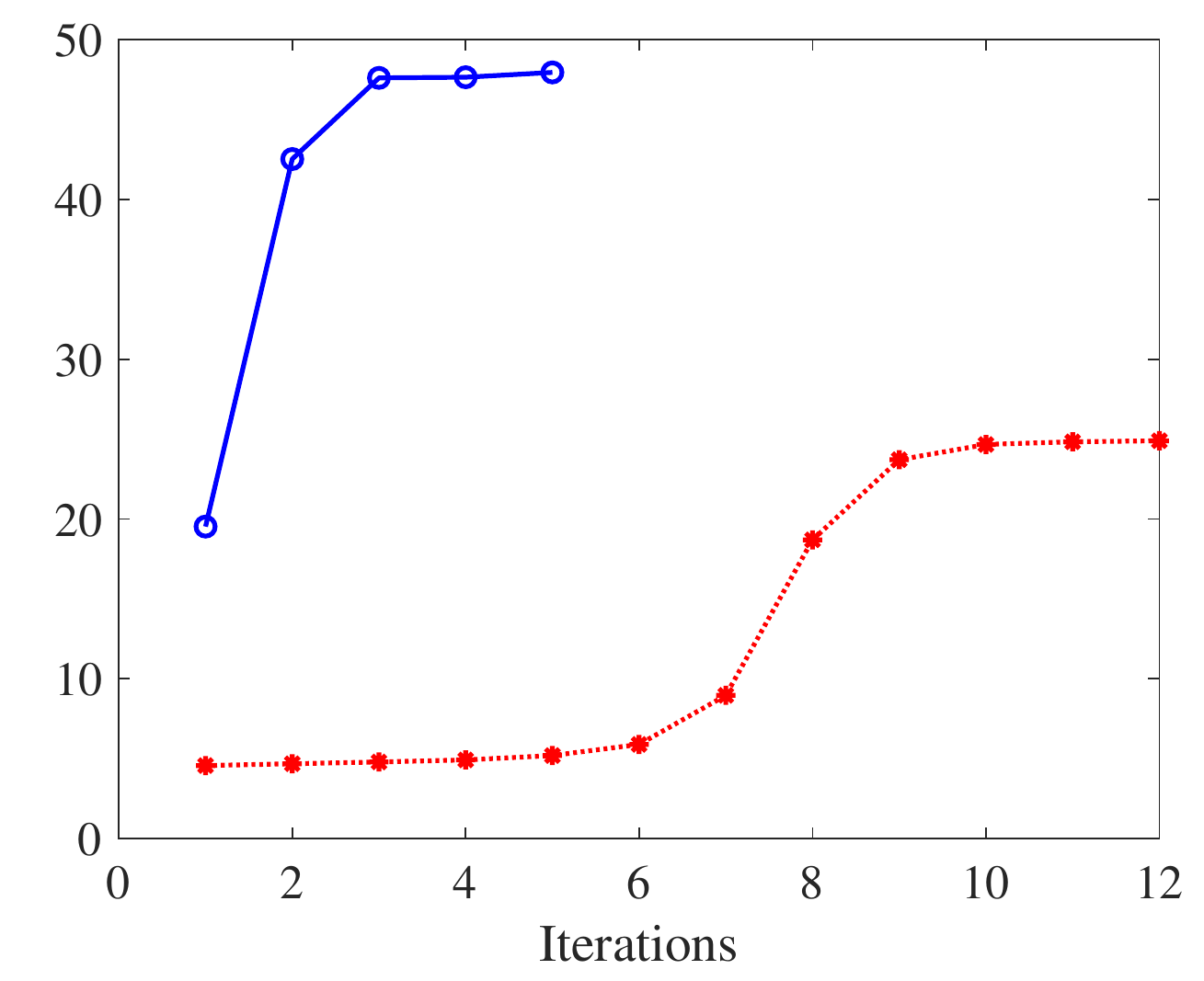} %,height=4.0cm
		\caption{Proposed}
		\label{fig:valuePCAKM}
	\end{subfigure}
	\caption{The optimal value versus the times of iteration on two datasets of (a) LDAKM, (b) LDAGMM, and (c) our proposed method.}
	\label{fig:valueObj}
\end{figure*}
\subsubsection{Monotonicity Experiments}
Finally, in order to verify the monotonicity of our proposed method, we graph the model objective value versus the times of iteration. The results on Difficult1\_noiseo5 and Difficult1\_noise20 are illustrated while other datasets with similar results are ignored. At the same time, we also compared the model-based spike sorting methods, LDAKM and LDAGMM. As shown in Fig.\ref{fig:valuePCAKM}, the proposed method always iterates towards the increasing objective values and converges after about $10$ iterations. Although LDAKM has a trend of increasing the objective values overall, its optimal value is not always increasing, since the objective function of the two-step solution for LDAKM is different~\cite{wang2019unsupervised}. Furthermore, LDAGMM adds GMM operations after each KM~\cite{keshtkaran2017noise}, making its objective value even iterate toward a decreasing direction. LDAKM and LDAGMM need more iterations (Fig.\ref{fig:valueLDAKM} and Fig.\ref{fig:valueLDAGMM}), and also more computational resources (Tables~\ref{tab:acc_without} and~\ref{tab:acc_with}).

%最后, 为了验证the Monotonicity of 我们提出的方法，相对于 every iteration time 的 objective value 图被记录下来。这里我们展示了在Difficult1_noiseo5和Difficult1_noise20的结果，在其他数据集的具有类似的结果而被忽略了。同时，我们也对比了基于模型的spikes sorting方法，LDAKM和LDAGMM。如图\ref{fig:valuePCAKM}所示，the proposed model始终朝着目标值增加的方向迭代，并在10次左右的迭代后收敛。相似地，LDAKM也朝着目标值增加的方向迭代，但由于LDAKM的俩步迭代的目标函数不一样，其optimal value并不是递增的。糟糕的是，LDAGMM为了得到所谓的更紧凑的结果在迭代KM时后增加了GMM操作，使得其目标值甚至朝着降低的方向迭代。LDAKM和LDAGMM不仅造成了迭代次数的增加（如图~\ref{fig:valueLDAKM}和~\ref{fig:valueLDAGMM}所示），同时造成更高的时间复杂度（如Table.~\ref{tab:acc_without}和~\ref{tab:acc_with}）.
%the two-step solution of LDAKM has different objective functions
%Compare with LDA-KM~\cite{keshtkaran2014unsupervised} and LDA-GMM~\cite{keshtkaran2017noise}: iteration time (x-coordinate) vs objective value (y-coordinate)
\subsection{Experiment on Real-world Dataset - ``HC1"}
In this section, the real-world dataset, HC1~\cite{harris2000accuracy}, is used and four automatic spike sorting methods, LE-KM~\cite{chah2011automated}, LPP-SCK~\cite{nguyen2014spike}, LDAKM~\cite{keshtkaran2014unsupervised}, LDAGMM~\cite{keshtkaran2017noise}, are compared. Similar as~\cite{ekanadham2014unified,keshtkaran2017noise}, we firstly filter the raw recorded signal using a highpass Butterworth filter at 250 Hz with order 50. Then, spikes were identified by exceeding a threshold that is assigned by~\cite{ekanadham2014unified}. For more detailed preprocessing about the spike filtering and detection, please refer to the open source~\cite{tetrode2014process}. Actually, after these preprocessings, all the ground truth spikes recorded by the intracellular electrode are detected. Then the spikes from the tetrode are concatenated into a vector for further automatic sorting. Finally, according to the existing ground truth intracellularly recorded from one neuron, we can obtain the corresponding FPR and FNR.
\begin{figure}[!htb]
	\centering
	\includegraphics[width=0.49\textwidth]{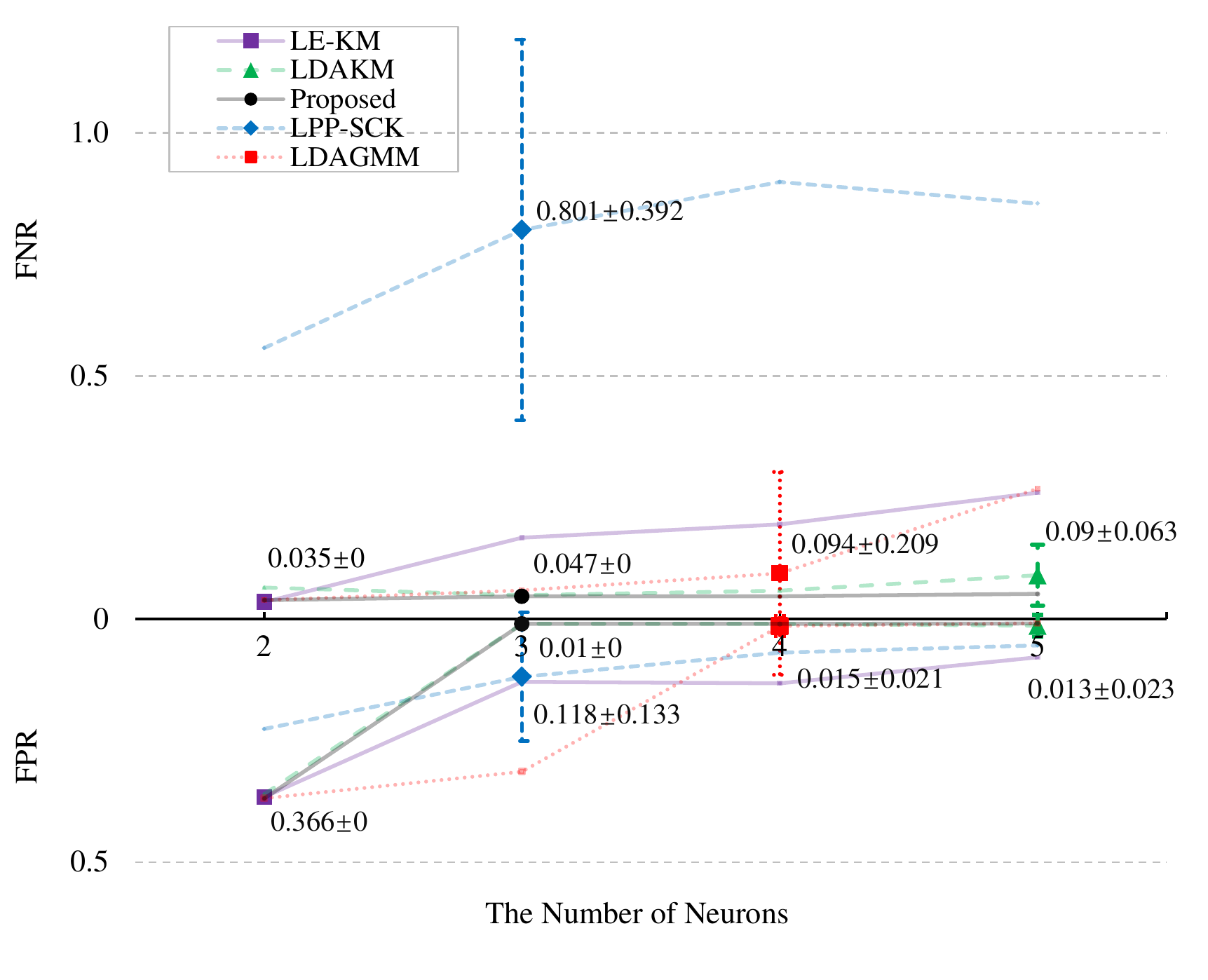}
	\caption{FPR and FNR obtained by several automatic spike sorting methods, i.e., LE-KM, LDAKM, LPP-SCK, LDAGMM and our proposed one on the HC1 datasets. Here, results with the estimated neurons from 2 to 5 are included, and among them, the bold dots indicate the actual estimated results (along with their standard deviation) for each method.}
	\label{fig:missVSfpr}
\end{figure}

In Figure \ref{fig:missVSfpr}, the results of five automatic spike sorting methods are shown when the estimated neurons ranges from $2$ to $5$. Among them, the actual estimated results are emphasised with bold dots and their numerical values are also accompanied. On one hand, it is easy to find that the results of LDAKM, LPP-SCK, and LDAGMM are unstable, especially the standard deviation of the FNR of LPP-SCK reaches $39.2\%$. On the other hand, LDAKM and our proposed method get better results for different estimated neurons compared with LE-KM, LPP-SCK and LDAGMM, although the actual estimated number is different between them (the proposed method is $3$ while LDAKM is $5$).

In addition to FPR and FNR, we also showed the mean waveforms of ground truth and the distribution of spikes obtained by different methods. Specifically, assuming the estimated number of neurons is $\tilde{c}$, when $\tilde{c}\leq 3$, the graphical dimension is $\tilde{c}-1$ while when $\tilde{c}>3$, the dimension is $3$. Meanwhile, the average waveform from each class of spike is also shown under the corresponding panel. As we can see from~Fig.\ref{fig:last}, the number of neurons and their waveforms obtained by various methods are different. Relatively speaking, the results of our proposed method are better since the graphical results show three distinctly different clusters and waveforms. In contrast, the results of LE-KM (Fig.\ref{fig:last-LEKM}) and LPP-LCK (Fig.\ref{fig:last-LPPLCK}) do not separate the spikes reasonably. With further analysing the average waveforms, we can see that the second (red) waveforms from our proposed method are more consistent with the ground truth compared with the third (yellow) waveforms from LDAKM and the second (red) waveforms from LDAGMM. Besides, it appears that the first (blue) and second (red), fourth (purple) and fifth (green) waveforms of LDAKM (Fig.\ref{fig:last-LDAKM}) are very similar, and we prefer to merge them together. %have more potentials to merge them together.
Similarly, for LDAGMM (Fig.\ref{fig:last-LDAGMM}), the third (yellow) and fourth (purple) waveforms are also very similar.

In summary, our proposed method exhibits excellent accuracy and stability in both synthetic and real-world datasets. %This not only confirms the validity of our proposed method but also confirms that it is more reasonable to model spike feature extraction and clustering together for spike sorting.
\begin{figure*}[!t]
	\centering
	\begin{subfigure}{0.3\linewidth}
		\centering
		\includegraphics[width=1\linewidth,height=4.0cm]{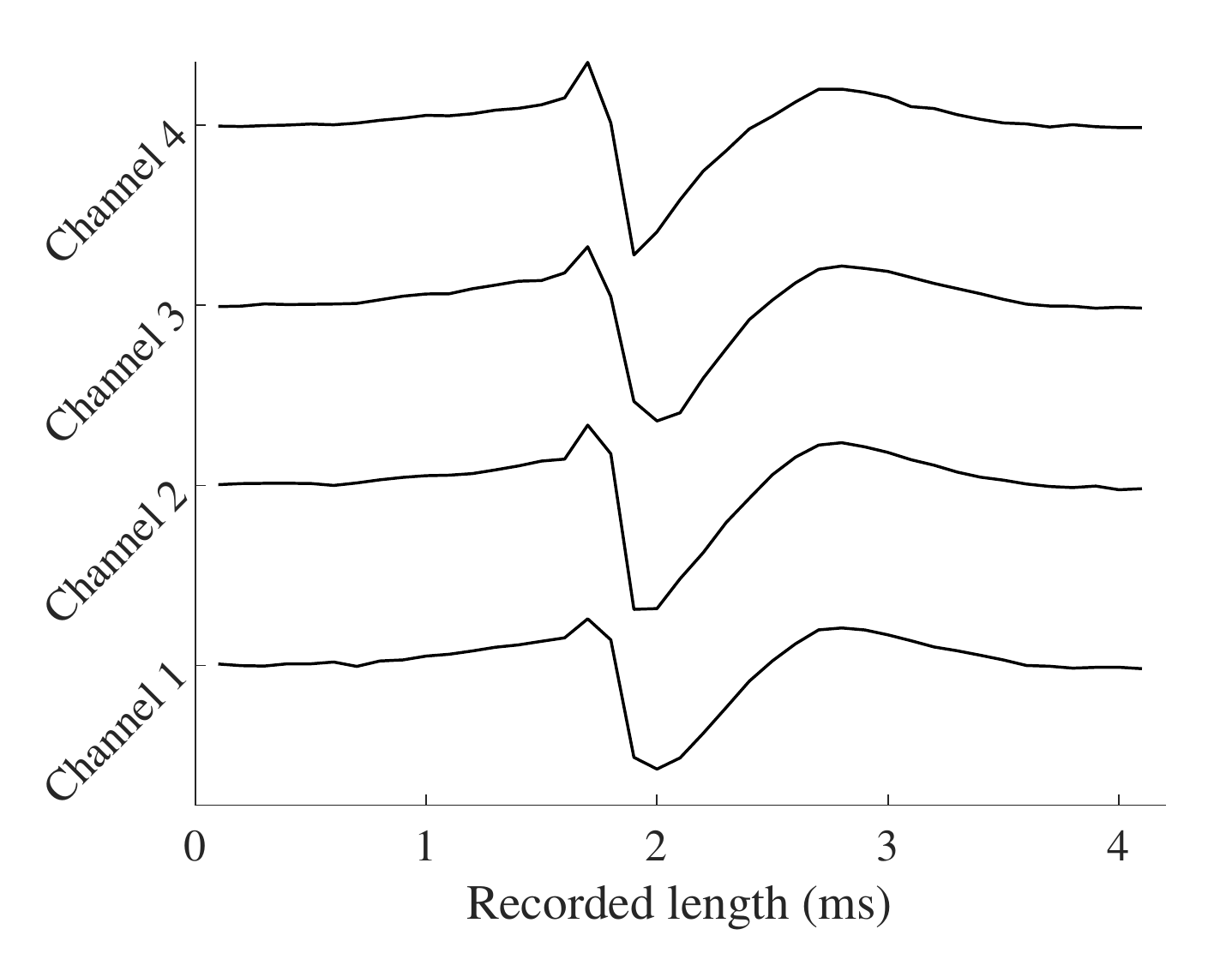}
		\caption{Ground truth}
		\label{fig:last-groundtruth}
	\end{subfigure}
	\begin{subfigure}{0.3\linewidth}
		\centering
		\includegraphics[width=1\linewidth,height=4.0cm]{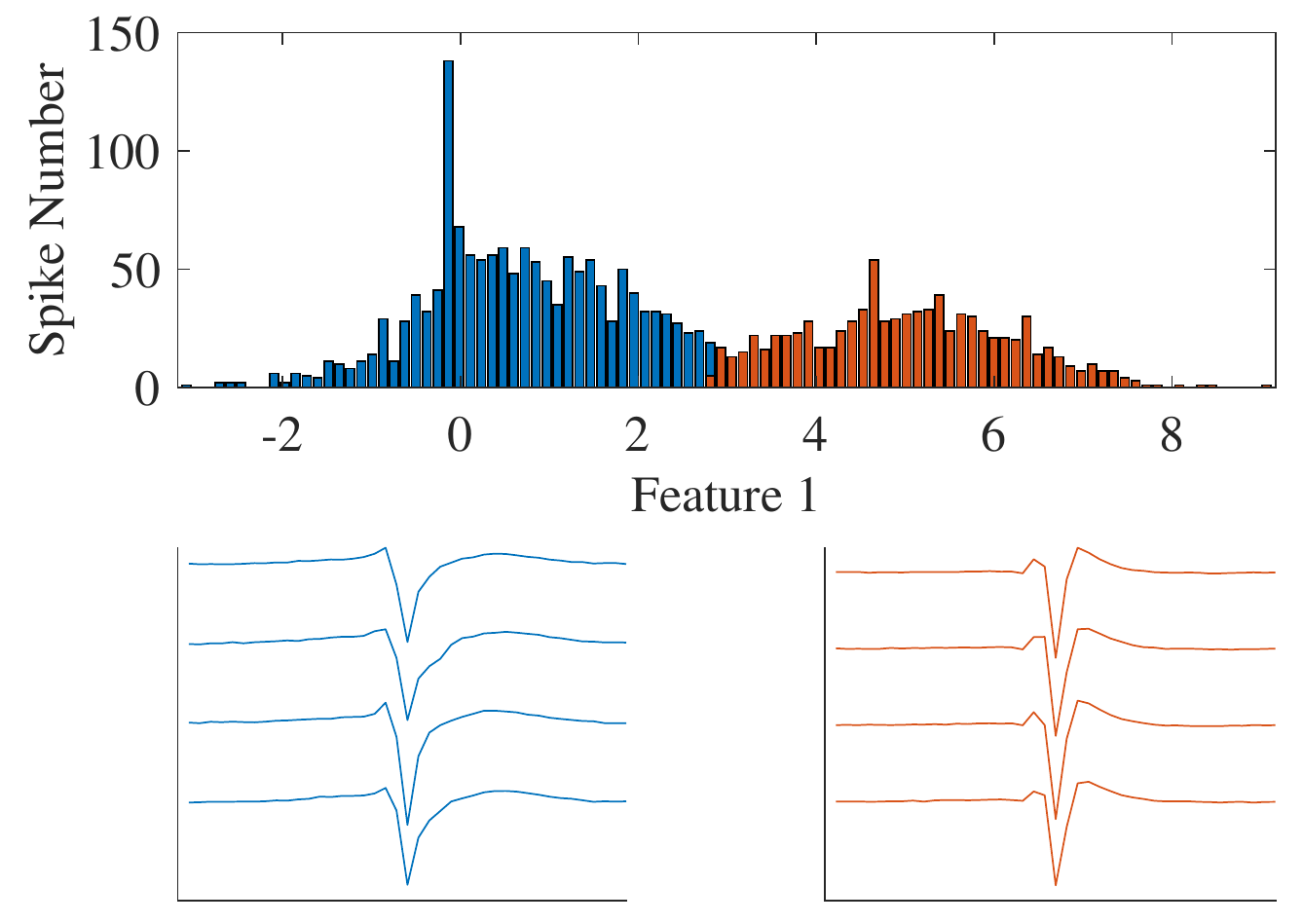}
		\caption{LE-KM}
		\label{fig:last-LEKM}
	\end{subfigure}
	\begin{subfigure}{0.3\linewidth}
		\centering
		\includegraphics[width=1\linewidth,height=4.0cm]{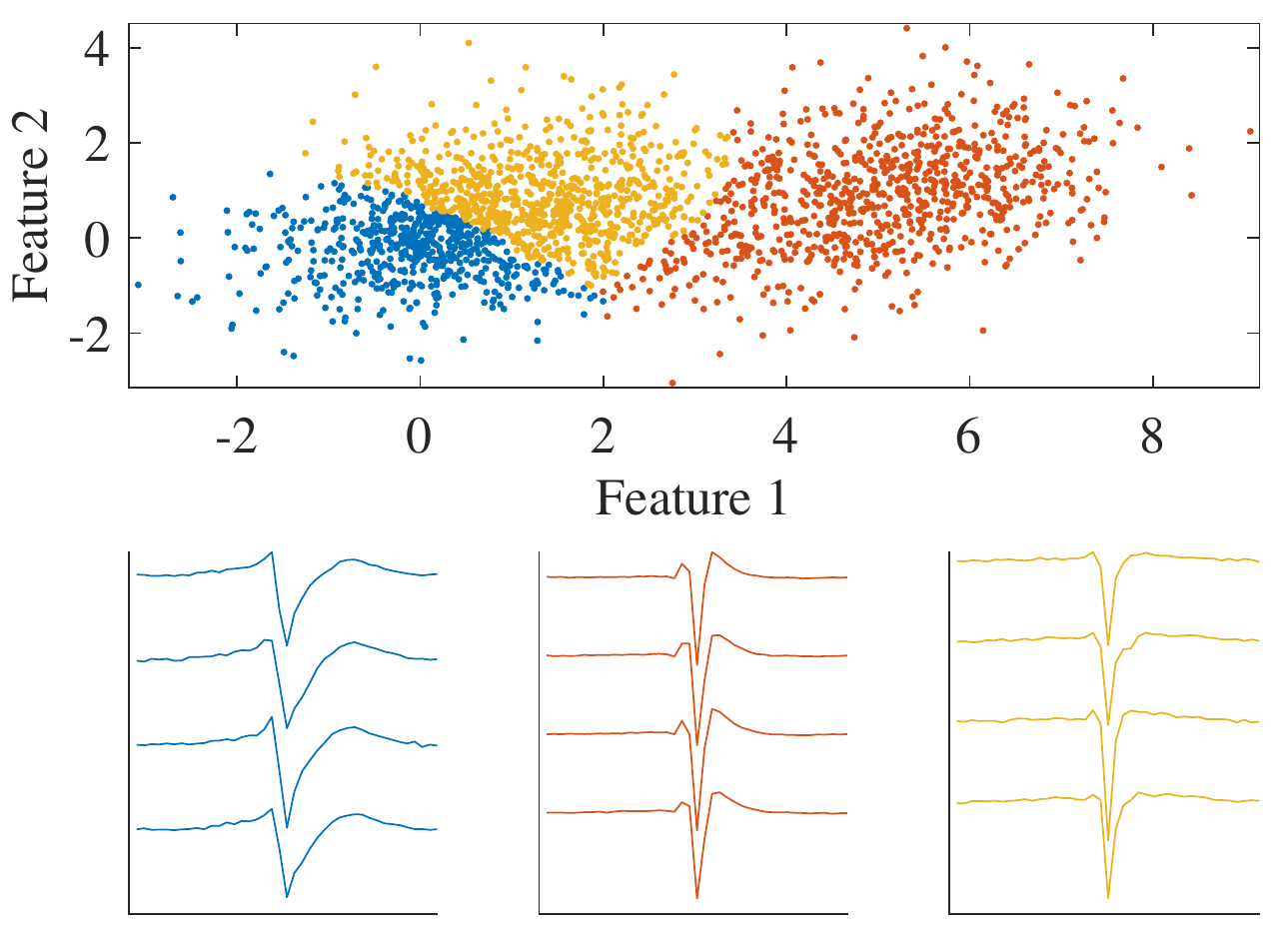}
		\caption{LPP-LCK}
		\label{fig:last-LPPLCK}
	\end{subfigure}
	\newline
	\begin{subfigure}{0.3\linewidth}
		\centering
		\includegraphics[width=1\linewidth,height=4.0cm]{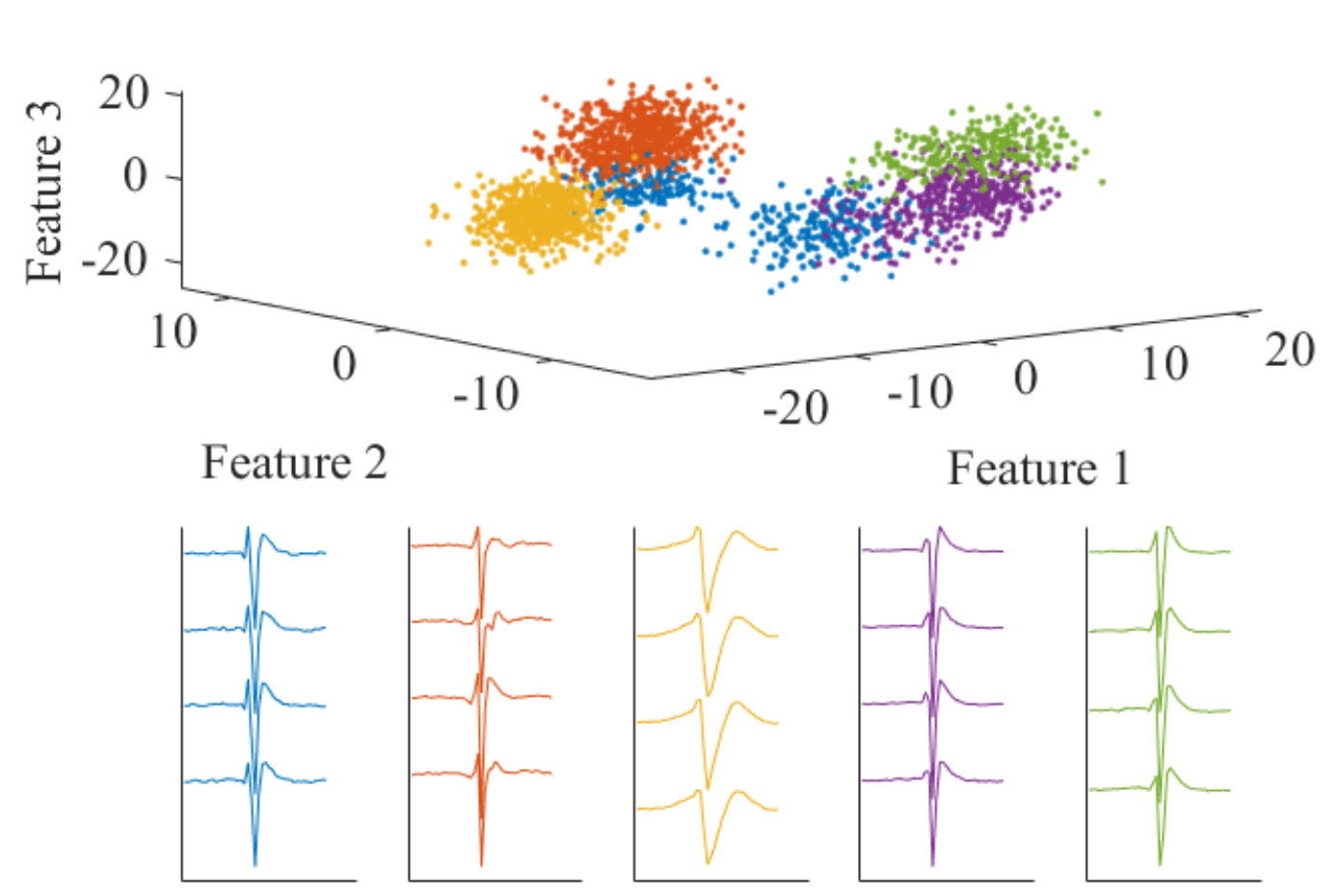}
		\caption{LDAKM}
		\label{fig:last-LDAKM}
	\end{subfigure}
	\begin{subfigure}{0.3\linewidth}
		\centering
		\includegraphics[width=1\linewidth,height=4.0cm]{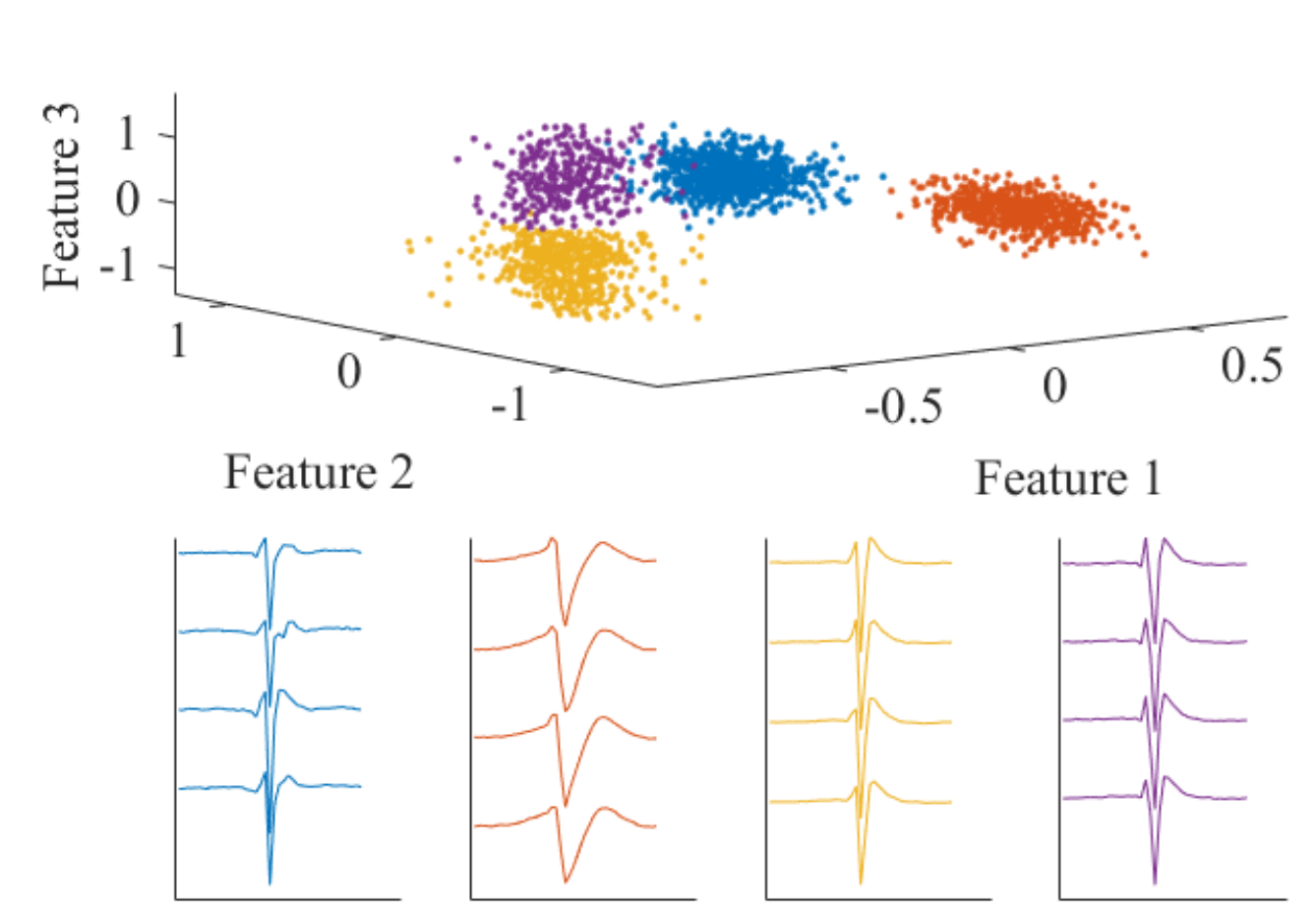}
		\caption{LDAGMM}
		\label{fig:last-LDAGMM}
	\end{subfigure}
	\begin{subfigure}{0.3\linewidth}
		\centering
		\includegraphics[width=1\linewidth,height=4.0cm]{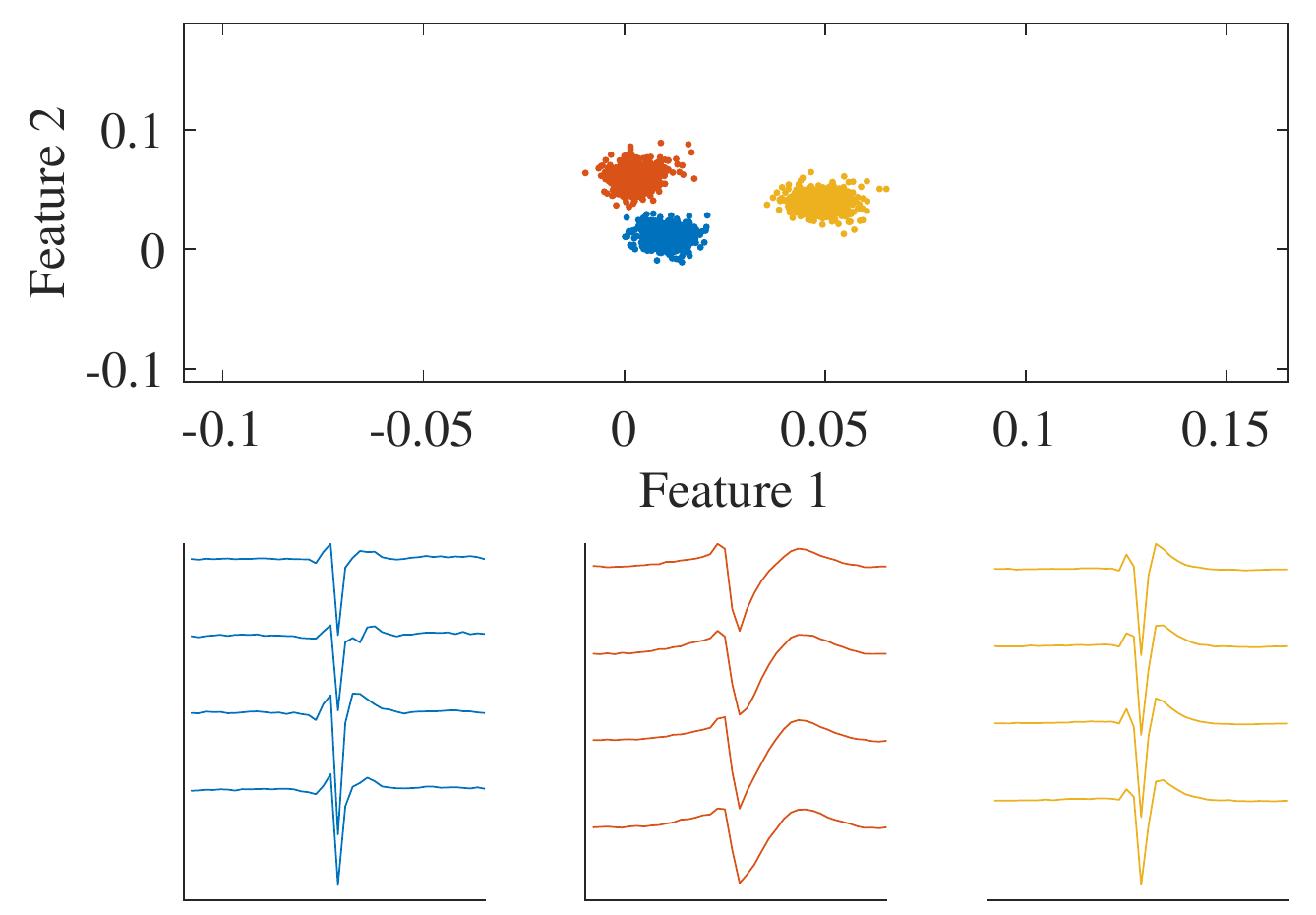}
		\caption{Proposed}
		\label{fig:last-Proposed}
	\end{subfigure}
	\caption{The mean waveform of ground truth is shown in (a), along with the low-dimensional distributions of spikes and their mean waveforms of each cluster from several automatic spike sorting schemes: (b) LE-KM, (c) LPP-LCK, (d) LDAKM, (e) LDAGMM, and (f) our proposed method.}
	\label{fig:last}
\end{figure*}

%% file: 5.Discussion.tex
\section{Conclusions}\label{sec:dis}
% This paper studied a unified optimisation method of feature extraction and clustering for spike sorting, which is solved in PCA and K-means. The computational complexity of this model is linear to the number of spikes. Besides, one efficient and effective automatic spike sorting method is proposed after integrated the existing estimated indices strategies. This proposed method yields more satisfactory and stable results compared with several similarly spike sorting methods on both synthetic and real datasets. %The method has the limitation of 

%%% revised by Yana (just for reference)
This paper studied a optimisation problem for spike sorting which unified the feature extraction and clustering. It was basically solved in an integrated framework with PCA and K-means. The computational complexity of this model is linear to the number of spikes, which is less than existing methods. In particular, we proposed an efficient and effective automatic spike sorting method after integrating existing strategies to estimate the number of neurons. The yielded satisfactory and stable results on both synthetic and real datasets verified the efficacy of our proposed method.

%% file: 6.Appendix.tex
\appendices
\section{Proofs} \label{app:A}
\textit{Proof (of Proposition \ref{prop:updateW}):} 
Let $S_1=XHX^T$ and $S_2=XH\bracket{I-G\bracket{G^TG}^{-1}G^T}HX^T$, we can simplify \eqref{equ:updateW} as,
\begin{equation}\label{equ:appB1}
\begin{split}
\max_{W}& \ \ \tr\curlybrace{\frac{W^TS_1W}{W^TS_2W}},\\
	s.t.,& \ \ W^TW=I.
\end{split}  \tag{{\ref{app:A}{.1}}}
\end{equation}
Because $S_1$ and $S_2$ are symmetric matrix, there exists a $d\times d$ nonsingular square matrix $A$ so that,
\begin{equation}\label{equ:appB2}
\begin{split}
	A^TW^TS_2WA=I_d, \ \ \mathrm{and} \ \ A^TW^TS_1WA=M_d,
\end{split}\tag{{\ref{app:A}{.2}}}
\end{equation}
where $M_d$ is a $d\times d$ diagonalised matrix. According to \eqref{equ:app3}, \eqref{equ:app4}, \eqref{equ:appB1}, and \eqref{equ:appB2}, we have,
\begin{align*}
	S_2^{-1}S_1W&=W\bracket{W^TS_2W}^{-1}\bracket{W^TS_1W} \nonumber\\
				&=W\bracket{\bracket{A^T}^{-1}I_dA^{-1}}^{-1}\bracket{\bracket{A^T}^{-1}M_dA^{-1}}\nonumber\\
				&=WAM_dA^{-1},
\end{align*}
that is, $S_2^{-1}S_1WA=WAM_d$. Besides, 
\begin{equation*}
\begin{split}
	&\tr\curlybrace{\frac{W^TS_1W}{W^TS_2W}} \\
	&=\tr\curlybrace{\bracket{A^TW^TS_2WA}^{-1}\bracket{A^TW^TS_1WA}},\\	
	&=\tr\curlybrace{\frac{A^TW^TS_1WA}{A^TW^TS_2WA}}=\tr\curlybrace{M_d}.
\end{split}
\end{equation*}
Therefore, the optimal value of \eqref{equ:appB1} can be determined by the $m$ largest eigenvalues of $M_d$. Also the corresponding eigenvectors in $WA$ are the optimal solution of \eqref{equ:appB1}. Certainly such eigenvalues and eigenvectors could be deduced from generalised eigenvalue decomposition on $S_2^{-1}S_1$, so they also satisfy the decomposition \eqref{equ:pro} and Proposition \ref{prop:updateW} is proved.

\section{Main mathematical Properties}\label{app:B}
The following equations
%\eqref{equ:app1} and \eqref{equ:app2}  deduced
are used in the Appendix~\ref{app:A}.
\begin{flalign} \label{equ:app1}
    \frac{d}{dX}\tr\curlybrace{X^TYX}=\left(Y+Y^T\right)X \tag{{\ref{app:B}{.1}}} &&
\end{flalign}
\begin{flalign} \label{equ:app2}
    &\frac{d}{dX}\tr\curlybrace{\frac{Y}{X^TSX}} \nonumber\\
    &=-SX\left(X^TSX\right)^{-1}\bracket{Y+Y^T}\bracket{X^TSX}^{-1} \tag{{\ref{app:B}{.2}}} &&
\end{flalign}
\begin{flalign} \label{equ:app3}
    &\frac{d}{dX}\tr\curlybrace{\frac{X^TS_1X}{X^TS_2X}} \nonumber \\
    &=\left.\frac{d}{dX_1}\tr\curlybrace{\frac{X^TS_1X}{X_1^TS_2X_1}}\right|_{X_1=X}\nonumber\\
    & \ \ \ +\left.\frac{d}{dX_2}\tr\curlybrace{\frac{X_2^TS_1X_2}{X^TS_2X}}\right|_{X_2=X} \nonumber\\
    &=-2S_2X\bracket{X^TS_2X}^{-1}\bracket{X^TS_1X}\bracket{X^TS_2X}^{-1}\nonumber\\
    & \ \ \ +2S_1X\bracket{X^TS_2X}^{-1} \tag{{\ref{app:B}{.3}}} &&
\end{flalign}
%The above deduced equations \eqref{equ:app1} and \eqref{equ:app2} are used in the Appendix~\ref{app:A}.
When \eqref{equ:app3} is set to be equal to zero, we have
\begin{flalign} \label{equ:app4}
    S_2^{-1}S_1X=X\bracket{X^TS_2X}^{-1}\bracket{X^TS_1X}. \tag{{\ref{app:B}{.4}}} &&
\end{flalign}

%% file: ms.bbl
\begin{thebibliography}{10}
\providecommand{\url}[1]{#1}
\csname url@samestyle\endcsname
\providecommand{\newblock}{\relax}
\providecommand{\bibinfo}[2]{#2}
\providecommand{\BIBentrySTDinterwordspacing}{\spaceskip=0pt\relax}
\providecommand{\BIBentryALTinterwordstretchfactor}{4}
\providecommand{\BIBentryALTinterwordspacing}{\spaceskip=\fontdimen2\font plus
\BIBentryALTinterwordstretchfactor\fontdimen3\font minus
  \fontdimen4\font\relax}
\providecommand{\BIBforeignlanguage}[2]{{%
\expandafter\ifx\csname l@#1\endcsname\relax
\typeout{** WARNING: IEEEtran.bst: No hyphenation pattern has been}%
\typeout{** loaded for the language `#1'. Using the pattern for}%
\typeout{** the default language instead.}%
\else
\language=\csname l@#1\endcsname
\fi
#2}}
\providecommand{\BIBdecl}{\relax}
\BIBdecl

\bibitem{lewicki1998review}
M.~S. Lewicki, ``A review of methods for spike sorting: the detection and
  classification of neural action potentials,'' \emph{Network: Computation in
  Neural Systems}, vol.~9, no.~4, pp. R53--R78, 1998.

\bibitem{mokri2017sorting}
Y.~Mokri, R.~F. Salazar, B.~Goodell, J.~Baker, C.~M. Gray, and S.-C. Yen,
  ``Sorting overlapping spike waveforms from electrode and tetrode
  recordings,'' \emph{Frontiers in neuroinformatics}, vol.~11, p.~53, 2017.

\bibitem{harris2000accuracy}
K.~D. Harris, D.~A. Henze, J.~Csicsvari, H.~Hirase, and G.~Buzsaki, ``Accuracy
  of tetrode spike separation as determined by simultaneous intracellular and
  extracellular measurements,'' \emph{Journal of neurophysiology}, vol.~84,
  no.~1, pp. 401--414, 2000.

\bibitem{carlson2019continuing}
D.~Carlson and L.~Carin, ``Continuing progress of spike sorting in the era of
  big data,'' \emph{Current opinion in neurobiology}, vol.~55, pp. 90--96,
  2019.

\bibitem{rey2015past}
H.~G. Rey, C.~Pedreira, and R.~Q. Quiroga, ``Past, present and future of spike
  sorting techniques,'' \emph{Brain research bulletin}, vol. 119, pp. 106--117,
  2015.

\bibitem{gibson2011spike}
S.~Gibson, J.~W. Judy, and D.~Markovi{\'c}, ``Spike sorting: The first step in
  decoding the brain: The first step in decoding the brain,'' \emph{IEEE Signal
  processing magazine}, vol.~29, no.~1, pp. 124--143, 2011.

\bibitem{franke2015bayes}
F.~Franke, R.~Q. Quiroga, A.~Hierlemann, and K.~Obermayer, ``Bayes optimal
  template matching for spike sorting--combining fisher discriminant analysis
  with optimal filtering,'' \emph{Journal of computational neuroscience},
  vol.~38, no.~3, pp. 439--459, 2015.

\bibitem{linderman2007signal}
M.~D. Linderman, G.~Santhanam, C.~T. Kemere, V.~Gilja, S.~O'Driscoll, M.~Y.
  Byron, A.~Afshar, S.~I. Ryu, K.~V. Shenoy, and T.~H. Meng, ``Signal
  processing challenges for neural prostheses,'' \emph{IEEE Signal Processing
  Magazine}, vol.~25, no.~1, pp. 18--28, 2007.

\bibitem{sanchez2007technology}
J.~C. Sanchez, J.~C. Principe, T.~Nishida, R.~Bashirullah, J.~G. Harris, and
  J.~A. Fortes, ``Technology and signal processing for brain-machine
  interfaces,'' \emph{IEEE Signal Processing Magazine}, vol.~25, no.~1, pp.
  29--40, 2007.

\bibitem{mahallati2019cluster}
S.~Mahallati, J.~C. Bezdek, M.~R. Popovic, and T.~A. Valiante, ``Cluster
  tendency assessment in neuronal spike data,'' \emph{PloS one}, vol.~14,
  no.~11, p. e0224547, 2019.

\bibitem{adamos2008performance}
D.~A. Adamos, E.~K. Kosmidis, and G.~Theophilidis, ``Performance evaluation of
  pca-based spike sorting algorithms,'' \emph{Computer methods and programs in
  biomedicine}, vol.~91, no.~3, pp. 232--244, 2008.

\bibitem{yang2013robust}
C.~Yang, Y.~Yuan, and J.~Si, ``Robust spike classification based on frequency
  domain neural waveform features,'' \emph{Journal of neural engineering},
  vol.~10, no.~6, p. 066015, 2013.

\bibitem{hulata2002method}
E.~Hulata, R.~Segev, and E.~Ben-Jacob, ``A method for spike sorting and
  detection based on wavelet packets and shannon's mutual information,''
  \emph{Journal of neuroscience methods}, vol. 117, no.~1, pp. 1--12, 2002.

\bibitem{chaure2018novel}
F.~J. Chaure, H.~G. Rey, and R.~Quian~Quiroga, ``A novel and fully automatic
  spike-sorting implementation with variable number of features,''
  \emph{Journal of neurophysiology}, vol. 120, no.~4, pp. 1859--1871, 2018.

\bibitem{nguyen2014spike}
T.~Nguyen, A.~Khosravi, D.~Creighton, and S.~Nahavandi, ``Spike sorting using
  locality preserving projection with gap statistics and landmark-based
  spectral clustering,'' \emph{Journal of neuroscience methods}, vol. 238, pp.
  43--53, 2014.

\bibitem{chah2011automated}
E.~Chah, V.~Hok, A.~Della-Chiesa, J.~Miller, S.~O'Mara, and R.~Reilly,
  ``Automated spike sorting algorithm based on laplacian eigenmaps and k-means
  clustering,'' \emph{Journal of neural engineering}, vol.~8, no.~1, p. 016006,
  2011.

\bibitem{ghanbari2010graph}
Y.~Ghanbari, P.~E. Papamichalis, and L.~Spence, ``Graph-laplacian features for
  neural waveform classification,'' \emph{IEEE transactions on biomedical
  engineering}, vol.~58, no.~5, pp. 1365--1372, 2010.

\bibitem{keshtkaran2017noise}
M.~R. Keshtkaran and Z.~Yang, ``Noise-robust unsupervised spike sorting based
  on discriminative subspace learning with outlier handling,'' \emph{Journal of
  neural engineering}, vol.~14, no.~3, p. 036003, 2017.

\bibitem{huang2020spike}
L.~Huang, B.~W.-K. Ling, Y.~Zeng, and L.~Gan, ``Spike sorting based on low-rank
  and sparse representation,'' in \emph{2020 IEEE International Conference on
  Multimedia and Expo (ICME)}.\hskip 1em plus 0.5em minus 0.4em\relax IEEE,
  2020, pp. 1--6.

\bibitem{souza2019spike}
B.~C. Souza, V.~Lopes-dos Santos, J.~Bacelo, and A.~B. Tort, ``Spike sorting
  with gaussian mixture models,'' \emph{Scientific reports}, vol.~9, no.~1, pp.
  1--14, 2019.

\bibitem{keshtkaran2014unsupervised}
M.~R. Keshtkaran and Z.~Yang, ``Unsupervised spike sorting based on
  discriminative subspace learning,'' in \emph{2014 36th Annual International
  Conference of the IEEE Engineering in Medicine and Biology Society}.\hskip
  1em plus 0.5em minus 0.4em\relax IEEE, 2014, pp. 3784--3788.

\bibitem{ding2007adaptive}
C.~Ding and T.~Li, ``Adaptive dimension reduction using discriminant analysis
  and k-means clustering,'' in \emph{Proceedings of the 24th international
  conference on Machine learning}, 2007, pp. 521--528.

\bibitem{jia2009trace}
Y.~Jia, F.~Nie, and C.~Zhang, ``Trace ratio problem revisited,'' \emph{IEEE
  Transactions on Neural Networks}, vol.~20, no.~4, pp. 729--735, 2009.

\bibitem{wang2019unsupervised}
F.~Wang, Q.~Wang, F.~Nie, Z.~Li, W.~Yu, and R.~Wang, ``Unsupervised linear
  discriminant analysis for jointly clustering and subspace learning,''
  \emph{IEEE Transactions on Knowledge and Data Engineering}, 2019.

\bibitem{he2012alternating}
B.~He, M.~Tao, and X.~Yuan, ``Alternating direction method with gaussian back
  substitution for separable convex programming,'' \emph{SIAM Journal on
  Optimization}, vol.~22, no.~2, pp. 313--340, 2012.

\bibitem{ekanadham2014unified}
C.~Ekanadham, D.~Tranchina, and E.~P. Simoncelli, ``A unified framework and
  method for automatic neural spike identification,'' \emph{Journal of
  neuroscience methods}, vol. 222, pp. 47--55, 2014.

\bibitem{hou2014discriminative}
C.~Hou, F.~Nie, D.~Yi, and D.~Tao, ``Discriminative embedded clustering: A
  framework for grouping high-dimensional data,'' \emph{IEEE transactions on
  neural networks and learning systems}, vol.~26, no.~6, pp. 1287--1299, 2015.

\bibitem{de2006discriminative}
F.~De~la Torre and T.~Kanade, ``Discriminative cluster analysis,'' in
  \emph{Proceedings of the 23rd international conference on Machine learning},
  2006, pp. 241--248.

\bibitem{xu2015pca}
Q.~Xu, C.~Ding, J.~Liu, and B.~Luo, ``Pca-guided search for k-means,''
  \emph{Pattern Recognition Letters}, vol.~54, pp. 50--55, 2015.

\bibitem{fukunaga2013introduction}
K.~Fukunaga, \emph{Introduction to statistical pattern recognition}.\hskip 1em
  plus 0.5em minus 0.4em\relax Elsevier, 2013.

\bibitem{arthur2006k}
D.~Arthur and S.~Vassilvitskii, ``k-means++: The advantages of careful
  seeding,'' Stanford, Tech. Rep., 2006.

\bibitem{rossant2016spike}
C.~Rossant, S.~N. Kadir, D.~F. Goodman, J.~Schulman, M.~L. Hunter, A.~B.
  Saleem, A.~Grosmark, M.~Belluscio, G.~H. Denfield, A.~S. Ecker \emph{et~al.},
  ``Spike sorting for large, dense electrode arrays,'' \emph{Nature
  neuroscience}, vol.~19, no.~4, pp. 634--641, 2016.

\bibitem{nguyen2015automatic}
T.~Nguyen, A.~Bhatti, A.~Khosravi, S.~Haggag, D.~Creighton, and S.~Nahavandi,
  ``Automatic spike sorting by unsupervised clustering with diffusion maps and
  silhouettes,'' \emph{Neurocomputing}, vol. 153, pp. 199--210, 2015.

\bibitem{zhang2018review}
J.~Zhang, T.~Nguyen, S.~Cogill, A.~Bhatti, L.~Luo, S.~Yang, and S.~Nahavandi,
  ``A review on cluster estimation methods and their application to neural
  spike data,'' \emph{Journal of neural engineering}, vol.~15, no.~3, p.
  031003, 2018.

\bibitem{quiroga2004unsupervised}
R.~Q. Quiroga, Z.~Nadasdy, and Y.~Ben-Shaul, ``Unsupervised spike detection and
  sorting with wavelets and superparamagnetic clustering,'' \emph{Neural
  computation}, vol.~16, no.~8, pp. 1661--1687, 2004.

\bibitem{wehr1999simultaneous}
M.~Wehr, J.~S. Pezaris, and M.~Sahani, ``Simultaneous paired intracellular and
  tetrode recordings for evaluating the performance of spike sorting
  algorithms,'' \emph{Neurocomputing}, vol.~26, pp. 1061--1068, 1999.

\bibitem{chen2011large}
X.~Chen and D.~Cai, ``Large scale spectral clustering with landmark-based
  representation,'' in \emph{Twenty-fifth AAAI conference on artificial
  intelligence}.\hskip 1em plus 0.5em minus 0.4em\relax Citeseer, 2011.

\bibitem{lefebvre2016recent}
B.~Lefebvre, P.~Yger, and O.~Marre, ``Recent progress in multi-electrode spike
  sorting methods,'' \emph{Journal of Physiology-Paris}, vol. 110, no.~4, pp.
  327--335, 2016.

\bibitem{pedreira2012many}
C.~Pedreira, J.~Martinez, M.~J. Ison, and R.~Q. Quiroga, ``How many neurons can
  we see with current spike sorting algorithms?'' \emph{Journal of neuroscience
  methods}, vol. 211, no.~1, pp. 58--65, 2012.

\bibitem{ldagmm2017code}
K.~Mohammad~Reza, ``Lda-gmm\_spikesort,''
  \url{https://github.com/mrezak/LDA-GMM_SpikeSort}, 2017, accessed:
  2019-07-29.

\bibitem{tetrode2014process}
E.~P. Simoncelli and P.~H. Li, ``Cbpspikesortdemo,''
  \url{https://github.com/chinasaur/CBPSpikesortDemo}, 2014, accessed:
  2019-03-01.

\end{thebibliography}
